\documentclass[11pt]{article}
\usepackage{graphicx,amsmath,amsfonts,amssymb,bm,hyperref,url,breakurl,epsfig,epsf,color,fullpage,MnSymbol,mathbbol,fmtcount,algorithmic,algorithm,semtrans,cite,caption,subcaption,
  multirow}

\usepackage{caption}
\usepackage[bottom,hang,flushmargin]{footmisc}
\usepackage{comment}
\usepackage[bottom]{footmisc}

\usepackage{hyperref}
\definecolor{darkred}{RGB}{150,0,0}
\definecolor{darkgreen}{RGB}{0,150,0}
\definecolor{darkblue}{RGB}{0,0,200}
\hypersetup{colorlinks=true, linkcolor=darkred, citecolor=darkgreen, urlcolor=darkblue}

\numberwithin{equation}{section}

\def \endprf{\hfill {\vrule height6pt width6pt depth0pt}\medskip}
\newenvironment{proof}{\noindent {\bf Proof} }{\endprf\par}

\newcommand{\EE}{\mathbb{E}}

\newcommand{\by}{\mathbf{y}}

\newcommand{\bv}{\mathbf{v}}

\newcommand{\cX}{\mathcal{X}}

\newcommand{\cV}{\mathcal{V}}
\newcommand{\cO}{\mathcal{O}}

\newcommand{\cR}{\mathcal{R}}
\newcommand{\cS}{\mathcal{S}}

\def\bbI{\mathbb{1}}

\newcommand{\eps}{\epsilon}
\newcommand{\tP}{\tilde{P}}
\newcommand{\tG}{\tilde{G}}
\newcommand{\bty}{\mathbf{\tilde{y}}}

\newcommand{\barA}{\bar{A}}

\newcommand{\pet}{p_{e_2}}

\newcommand{\wij}{w_{i,j}}
\newcommand{\wijl}{w_{i,j}^l}
\newcommand{\rijl}{r_{i,j}^l}
\newcommand{\tf}{\tilde{f}}
\newcommand{\bs}{\mathbf{s}}

\newenvironment{lbmatrix}[1]
  {\left[\array{@{}*{#1}{c}@{}}}
  {\endarray\right]}

\newtheorem{theorem}{\textbf{Theorem}}
\newtheorem{lemma}{\textbf{Lemma}}
\newtheorem{corollary}{\textbf{Corollary}}
\newtheorem{definition}{\textbf{Definition}}

\newtheorem{proposition}{\textbf{Proposition}}

\newtheorem{alg}{\textbf{Algorithm}}
\newtheorem{remark}{\textbf{Remark}}

\title{{\huge Spectral Alignment of Graphs}}
\date{}
\author{Soheil~Feizi$^{1}$, Gerald~Quon$^{2}$, Mariana~Recamonde-Mendoza$^{3}$, Muriel M{\'e}dard$^{4}$,\\ Manolis Kellis$^{4}$ and Ali~Jadbabaie$^{4}$\\
\vspace{-1.5cm}
\footnote{$^1$ Stanford University.}\\
\footnote{$^2$ University of California, Davis.}\\
\footnote{$^3$ Instituto de Informaìtica, Universidade Federal do Rio Grande do Sul, Porto Alegre, RS, Brazil.}\\
\footnote{$^4$ Massachusetts Institute of Technology (MIT).}}

\renewcommand\footnotemark{}

\begin{document}
\maketitle
\vspace{-1cm}
\begin{abstract}
Graph alignment refers to the problem of finding a bijective mapping across vertices of two graphs such that, if two nodes are connected in the first graph, their images are connected in the second graph. This problem arises in many fields such as computational biology, social sciences, and computer vision and is often cast as a quadratic assignment problem (QAP). Most standard graph alignment methods consider an optimization that maximizes the number of matches between the two graphs, ignoring the effect of mismatches. We propose a generalized graph alignment formulation that considers both matches and mismatches in a standard QAP formulation. This modification can have a major impact in aligning graphs with different sizes and heterogenous edge densities. Moreover, we propose two methods for solving the generalized graph alignment problem based on spectral decomposition of matrices. We compare the performance of proposed methods with some existing graph alignment algorithms including Natalie2, GHOST, IsoRank, NetAlign, Klau's approach as well as a semidefinite programming-based method over various synthetic and real graph models. Our proposed method based on simultaneous alignment of multiple eigenvectors leads to consistently good performance in different graph models. In particular, in the alignment of regular graph structures which is one of the most difficult graph alignment cases, our proposed method significantly outperforms other methods.
\end{abstract}

\section{Introduction}
The term graph alignment (or, network alignment) encompasses several distinct but related problem variants \cite{Trey2006modeling}. In general, graph alignment aims to find a bijective mapping across two (or more) graphs so that, if two nodes are connected in one graph, their images are also connected in the other graph(s). If such an exact alignment scheme exists, graph alignment can be simplified to the problem of {graph isomorphism} \cite{book1976graphbondy}. However, in general, an errorless alignment scheme may not be feasible. In such cases, graph alignment aims to find a mapping with the minimum error and/or the maximum overlap.

Graph alignment has a broad range of applications in systems biology, social sciences, computer vision, and linguistics. For instance, graph alignment has been used frequently as a comparative analysis tool in studying protein-protein interaction networks across different species \cite{isorank,isorankn,graemlin-flannick2006,graphmatching2009zaslavskiy,pathblast-kelley2004,networkblast-kalaev2008}. In computer vision, graph alignment has been used for image recognition by matching similar images \cite{visionconte2004thirty,visionschellewald2005probabilistic}. It has also been applied in ontology alignment to find relationships among different representations of a database \cite{ontology-lacoste2006word,ontology-melnik2002similarity}, and in user de-anonymization to infer user/sample identifications using similarity between datasets \cite{twitter-deanon}.

\begin{figure*}[t]
  \centering
      \includegraphics[width=14cm]{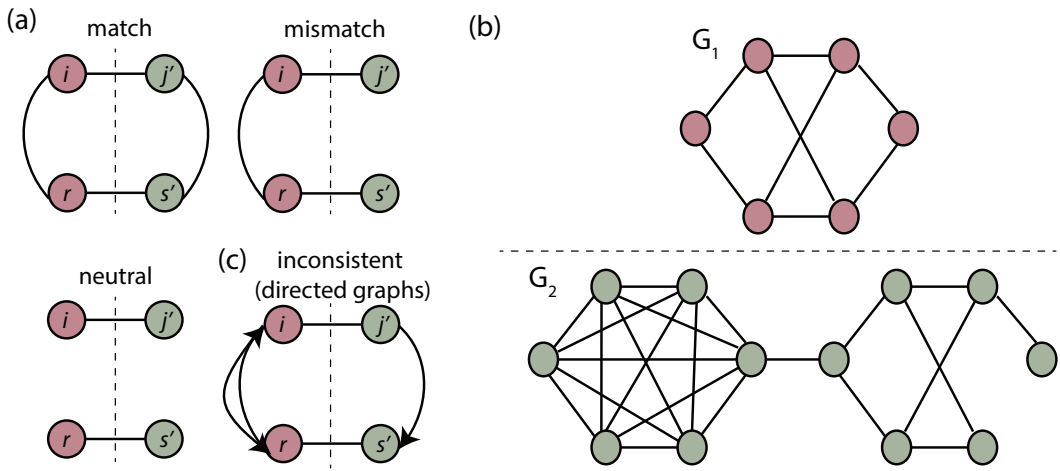}
  \caption{(a) An illustration of matched, mismatched, and neutral mappings for undirected graphs. (b) Example graphs to illustrate the effect of considering mismatches in the graph alignment formulation. (c) An illustration of inconsistent mappings for directed graphs where they are matches in one direction and mismatches in the other direction. }
  \label{fig:match-mismatch-illustration}
\end{figure*}

Here we study the graph alignment problem and make two main contributions. First, we propose a generalized formulation for the graph alignment optimization, and secondly we develop two graph alignment methods based on spectral decomposition of matrices. In the following we explain these contributions.

Let $G_1=(V_1,E_1)$ and $G_2=(V_2,E_2)$ be two graphs where $V_a$ and $E_a$ represent set of nodes and edges of graph $a=1,2$, respectively. By a slight abuse of notation, let $G_1$ and $G_2$ be their matrix representations where for $a=1,2$, $G_a(i,j)=1$ iff $(i,j)\in E_a$, and $G_a(i,j)=0$ otherwise. Suppose graph $a$ has $n_a$ nodes, i.e., $|V_a|=n_a$. Without loss of generality, we assume $n_1\leq n_2$. In the beginning, we assume graphs are undirected (i.e., matrices $G_1$ and $G_2$ are symmetric). We discuss the alignment of directed graphs, denoted by $G_1^{dir}$ and $G_2^{dir}$, in Section \ref{sec:eigenalign}.

Let $X$ be an $n_1\times n_2$ binary matrix where $X(i,j')=1$ means that node $i$ in graph $1$ is mapped (aligned) to node $j'$ in graph $2$. The pair $(i,j')$ is called a mapping edge across two graphs and is denoted by $i\leftrightarrow j'$. In the graph alignment setup, each node in one graph can be mapped to at most one node in the other graph, i.e., $\sum_{i} X(i,j')\leq 1$ for all $j'$, and similarly $\sum_{j'} X(i,j')\leq 1$ for all $i$. We also assume that there are no unaligned nodes in the graph with fewer nodes.

Matrix $X$ can map an edge in graph $G_1$ to an edge in graph $G_2$. These aligned edges are called {\it matches}. $X$ can map an edge in $G_1$ to a non-existing edge in $G_2$ and vice versa. These mapping pairs are called {\it mismatches}. Finally, $X$ can map a non-existing edge in $G_1$ to a non-existing edge in $G_2$. These pairs are called {\it neutrals}. Figure \ref{fig:match-mismatch-illustration}-a illustrates examples of matches, mismatches, and neutrals for simple graphs with two nodes. We have
\begin{align}\label{eq:match-mismatch-neutral-def}
&\text{\# of matches}=Tr \left(G_1 X G_2 X^T\right),\\
&\text{\# of mismatches}=Tr \left(G_1 X (\bbI-G_2) X^T+ G_1 X (\bbI-G_2) X^T\right),\nonumber\\
&\text{\# of neutrals}=Tr\left((\bbI-G_1) X (\bbI-G_2) X^T\right),\nonumber
\end{align}
where $\bbI$ represents a matrix of all ones and $Tr(.)$ is the trace operator. Most existing scalable graph alignment methods only consider maximizing the number of matches across two graphs while ignoring the number of resulting mismatches. This limitation can be critical particularly in cases where graphs have different sizes. We propose a generalized objective function for the graph alignment optimization as follows:
\begin{align}\label{eq:alignment-general-formulation}
\max_{X}~~ &s_1 (\#\text{ of matches})+s_2 (\#\text{ of neutrals})\\
&+ s_3 (\#\text{ of mismatches}),\nonumber
\end{align}
where $s_1$, $s_2$, and $s_3$ are scores assigned to matches, neutrals, and mismatches, respectively. We assume $s_1>s_2>s_3$. Considering $s_2=s_3=0$ results in ignoring effects of mismatches and neutrals. Substituting \eqref{eq:match-mismatch-neutral-def} in \eqref{eq:alignment-general-formulation}, we have the following equivalent optimization:
\begin{align}\label{eq:alignment-general-formulation-v2}
\max_{X}~~ Tr\left(G_1 X G_2 X^T\right)-\gamma \left(Tr\left(G_1 X \bbI X^T\right)+Tr\left(\bbI X G_2 X^T\right)\right),
\end{align}
where $\gamma=(s_2-s_3)/(s_1+s_2-2s_3)$ is the regularization parameter. Note $0\leq \gamma< 1/2$. If $s_2=s_3=0$, $\gamma=0$, while if $s_2\to s_1$ or $s_3\to-\infty$, $\gamma\to 1/2$. If $n_1=n_2$, $Tr\left(G_1 X \bbI X^T\right)$ and $Tr\left(\bbI X G_2 X^T\right)$ are equal to the number of edges in graphs $G_1$ and $G_2$, respectively. Thus these terms do not depend on $X$. However, if the number of nodes in $G_1$ and $G_2$ are different (say $n_1<n_2$), $Tr\left(\bbI X G_2 X^T\right)$ depends on $X$. Therefore, the regularization parameter $\gamma$ plays a role when the number of nodes in $G_1$ and $G_2$ are different. Note that in solving relaxations or approximations of optimization \eqref{eq:alignment-general-formulation-v2} when $X$ is no longer a permutation, $\gamma$ can have an effect even for the same size graphs.

To illustrate the effect of the regularization parameter, consider example graphs $G_1$ and $G_2$ illustrated in Figure \ref{fig:match-mismatch-illustration}-b.
Let $X_1$ and $X_2$ be mapping matrices that align nodes of $G_1$ to left and right subgraphs of $G_2$, respectively. The number of matches and mismatches caused by $X_1$ are $8$ and $7$, respectively. The number of matches and mismatches caused by $X_2$ are $7$ and $1$, respectively. If we ignore the effect of mismatches (i.e., $\gamma=0$ in \eqref{eq:alignment-general-formulation-v2}), $X_1$ leads to a larger graph alignment objective value compared to $X_2$. However, if $\gamma>1/6$, $X_2$ leads to a larger objective value compared to $X_1$. Note that maximizing matches while ignoring mismatches favors parts of the larger graph with a higher edge density.

It is important to note that the notion of mismatches has been considered in other alignment frameworks as well. For example \cite{visionschellewald2005probabilistic} considers aligning two images (modeled as graphs) knowing a pairwise similarity measure between nodes of the two graphs. Reference \cite{visionschellewald2005probabilistic} uses mismatch terms (ignoring matches) to incorporate relational structure terms in the alignment optimization. Our generalized graph alignment optimization \eqref{eq:alignment-general-formulation-v2} does not require having a similarity matrix between nodes of the two graphs and uses both match and mismatch information to compute the alignment matrix.

The objective function of optimization \eqref{eq:alignment-general-formulation-v2} is not in the standard form of a quadratic assignment problem (QAP) since it has three terms. It is straightforward to show that the following optimization is an equivalent formulation:
\begin{align}\label{eq:alignment-compact-formulation}
\max_{X}\quad Tr\left((G_1-\gamma\bbI) X (G_2-\gamma\bbI) X^T\right).
\end{align}
This optimization is a standard QAP \cite{quadratic-book-burkard2013} which is computationally challenging to solve. In the next section we explain our algorithmic contributions to compute a solution for this optimization based on spectral decomposition of functions of adjacency matrices.

Reference \cite{QAP-newlp-NPhard} shows that approximating a solution of maximum quadratic assignment problem within a factor better than $2^{\log^{1-\eps}n}$ is in general not feasible in polynomial time. However, owing to numerous applications of QAP in different areas, several algorithms have been designed to solve it approximately. Some methods use exact search approaches based on branch-and-bound \cite{branch-and-bound} and cutting plane \cite{cuttingplane}. These methods can only be applied to very small problem instances owing to their high computational complexity. Some methods attempt to solve the underlying QAP by linearizing the quadratic term and transforming the optimization into a mixed integer linear program (MILP) \cite{lawler1963quadratic,kaufman1978algorithm,frieze1983quadratic,adams1994improved}. In practice the very large number of introduced variables and constraints in linearization of the QAP objective function poses an obstacle for solving the resulting MILP efficiently. Some methods use convex relaxations of the QAP to compute a bound on its optimal value \cite{orthogonal-finke1987quadratic,projection-hadley1992new,convex-anstreicher2000lagrangian,convex-anstreicher2001solving,semidefinite-zhao1998}. The solutions provided by these methods may not be a feasible solution for the original quadratic assignment problem. Other methods to solve the graph alignment optimization include semidefinite \cite{splitting-sdp,semidefinite-zhao1998}, non-convex \cite{vogelstein2015fast}, or Lagrangian \cite{klau2009new,natalie,natalie2} relaxations, Bayesian inference \cite{Bayesiankolavr2012graphalignment}, message passing \cite{bayati2013message} or other heuristics \cite{isorank,isorankn,vision-spectral-main,vision-spectral-clustering,graphmatching2009zaslavskiy,kazemi2015growing,clark2015multiobjective,malod2015graal}. We will review these methods in Section \ref{sec:prior}. For more details about these methods, we refer readers to references \cite{quadratic-book-burkard2013,QAPsurvey-Elsevier,emmert2016fifty}. In particular \cite{emmert2016fifty} provides a recent review of graph alignment methods by distinguishing between methods for deterministic and random graphs.

Spectral inference methods have received significant attention in problems such as graph clustering  \cite{newman2006modularity,sussman,qin,athreya,saade} where the underlying mixed integer program is tightly approximated with an optimization whose optimizers can be computed efficiently. However, the use of spectral techniques in the graph alignment problem has been limited \cite{isorank,isorankn,vision-spectral-main,vision-spectral-clustering,schellewald2007evaluation,patro2012global}, partially owing to difficulty in connecting existing spectral graph alignment methods with relaxations of the underlying QAP. For example, \cite{isorank} computes an alignment across biological networks using the top eigenvector of a graph which encodes neighborhood similarities. Reference \cite{schellewald2007evaluation} uses a spectral relaxation of QAP to compute a probabilistic subgraph matching when the number of nodes of graphs are the same, while \cite{patro2012global} uses a heuristic multi-scale spectral signature of graphs to compute an alignment across them.

In this paper, we propose two spectral algorithms for solving the graph alignment optimization \eqref{eq:alignment-compact-formulation}, namely {\it EigenAlign} (EA), and {\it LowRankAlign} (LRA):

{\bf 1. EigenAlign (EA)} computes the leading eigenvector of a function of adjacency matrices followed by a maximum weight bipartite matching optimization. EigenAlign can be applied to both directed and undirected graphs. We prove that for Erd\H{o}s-R\'enyi graphs \cite{erdHos1961strength} and under some general conditions, EigenAlign is mean-field optimal \footnote{Finding an isomorphic mapping across asymptotically large Erd\H{o}s-R\'enyi graphs is a well studied problem and can be solved efficiently through canonical labeling \cite{isom-czajka2008improved}. Moreover Laszlo Babai has recently outlined his proof that the computational complexity of the general graph isomorphism problem is Quasipolynomial \cite{babai2016graph}. Note that in the graph alignment setup input graphs do not need to be isomorphic.}.

{\bf 2. LowRankAlign (LRA)} solves the graph alignment optimization by simultaneous alignment of eigenvectors of (transformations of) adjacency graphs, scaled by corresponding eigenvalues. LRA considers undirected graphs. LRA first solves the orthogonal relaxation of the underlying QAP using eigen decomposition of matrices. Then, it employs a rounding step as a projection in the direction of top eigenvectors of input matrices. We provide a bound on the performance of this projection step based on eigenvalues of input matrices and the orthogonal relaxation gap. Note that this rounding step is different than previously studied orthogonal projection, which has been shown to have a poor performance in practice \cite{schellewald2007evaluation}.

Through analytical performance characterization, simulations on several synthetic graphs, and real-data analysis, we show that our proposed graph alignment methods lead to improved performance compared to some existing graph alignment methods. Note that our proposed generalized graph alignment framework can also be adapted to some existing graph alignment packages. However, exploring this direction is beyond the scope of this article.

The rest of the paper is organized as follows. In Section \ref{sec:prior}, we review some existing graph alignment techniques and explain the relationship between graph alignment and graph isomorphisim. In Section \ref{sec:eigenalign}, we introduce the EigenAlign Algorithm and discuss its relationship with the underlying quadratic assignment problem. Moreover, we present the mean-field optimality of this method over random graphs, under some general conditions. In Section \ref{sec:lowrank}, we consider the trace formulation of the graph alignment optimization and introduce LowRankAlign. In Section \ref{sec:eval}, we compare performance of our method with some existing graph alignment methods over different synthetic graph structures. In Section \ref{sec:regulatory}, we use our graph alignment methods in comparative analysis of gene regulatory networks across different species.
\section{Review of Prior Work}\label{sec:prior}
Graph alignment problem \eqref{eq:alignment-compact-formulation} is an example of a QAP \cite{quadratic-book-burkard2013}. In the following we briefly summarize previous works by categorizing them into four groups and explain advantages and shortcomings of each. For more details on these methods we refer readers to references \cite{clark2014comparison,quadratic-book-burkard2013,QAPsurvey-Elsevier}.

{\bf 1. Exact search methods:} These methods provide a globally optimal solution for QAP. Examples of exact algorithms include methods based on branch-and-bound \cite{branch-and-bound} and cutting plane \cite{cuttingplane}. Owing to their high computational complexity, they can only be applied to very small problem instances.

{\bf 2. Linearizations:} These methods attempt to solve QAP by eliminating the quadratic term in the objective function, transforming it into a mixed integer linear program (MILP). An existing MILP solver is applied to find a solution for the relaxed problem. Examples of these methods are Lawler’s linearization \cite{lawler1963quadratic}, Kaufmann and Broeckx linearization \cite{kaufman1978algorithm}, Frieze and Yadegar linearization \cite{frieze1983quadratic}, and Adams and Johnson linearization \cite{adams1994improved}. These linearizations can provide bounds on the optimal value of the underlying QAP \cite{QAP-newlp-NPhard}. Moreover \cite{klau2009new,natalie,natalie2} use Lagrangian relaxations to compute a solution for the QAP. In general, linearization of the QAP objective function is achieved by introducing many new variables and new linear constraints. In practice, the very large number of introduced variables and constraints poses an obstacle for solving the resulting MILP efficiently.

{\bf 3. Semidefinite/convex relaxations:} These methods aim to compute a bound on the optimal value of the graph alignment optimization by considering the alignment matrix in the intersection of orthogonal and stochastic matrices. The provided solution by these methods may not be a feasible solution for the original quadratic assignment problem. Examples of these methods include orthogonal relaxations \cite{orthogonal-finke1987quadratic}, projected eigenvalue bounds \cite{projection-hadley1992new}, convex relaxations \cite{convex-anstreicher2000lagrangian,convex-anstreicher2001solving,semidefinite-zhao1998}, and matrix splittings \cite{splitting-sdp}. In particular, \cite{splitting-sdp} introduces a convex relaxation of the underlying graph alignment optimization based on matrix splitting which provides bounds on the optimal value of the underlying QAP. The proposed semidefinite programming (SDP) method provides a bound on the optimal value and additional steps are required to derive a feasible solution. Moreover, owing to its computational complexity, it can only be used to align small graphs \cite{splitting-sdp}.

In the computer vision literature, \cite{vision-spectral-main,vision-spectral-clustering} use spectral techniques to solve QAP approximately by inferring a cluster of assignments over the feature graph. Then, they use a greedy approach to reject assignments with low associations. Similarly, \cite{schellewald2007evaluation} uses a spectral relaxation of QAP to compute a probabilistic subgraph matching across images when the size of graphs are the same, while \cite{patro2012global} uses a heuristic multi-scale spectral signature of graphs to compute an alignment across them.

{\bf 4. Other methods:} There are several other techniques to solve graph alignment optimization approximately. Some methods use Bayesian framework \cite{Bayesiankolavr2012graphalignment}, or message passing \cite{bayati2013message}, or some other heuristics \cite{isorank,isorankn,graphmatching2009zaslavskiy}. In Section \ref{sec:eval}, we assess the performance of some of these graph alignment techniques through simulations.

Some graph alignment formulations aim to align paths \cite{pathblast-kelley2004} or subgraphs \cite{networkblast-kalaev2008,moduleali2009functionally,mohammadi2015triangular} across two (or multiple) graphs. The objective of these methods is different from the one of our graph alignment optimization where a bijective mapping across nodes of two graphs is desired according to a QAP. However solutions of these different methods may be related. For instance a bijective mapping across nodes of two graphs can provide information about conserved pathways and/or subgraphs across graphs, and vice versa.

The graph alignment formulation of \eqref{eq:alignment-compact-formulation} uses the structure of input graphs to find an alignment across their nodes. In practice, however, some other side information may be available such as node-node similarities. One way to incorporate such information in the formulation of \eqref{eq:alignment-compact-formulation} is to restrict the alignment across nodes of the two graphs whose similarities are greater than a threshold. This can be done by adding additional constraints to \eqref{eq:alignment-compact-formulation}. We will explain this in more detail in Section \ref{sec:eigenalign}.
\subsection{Graph Alignment and Graph Isomorphism}\label{sec:isomorphism}
The graph alignment optimization \eqref{eq:alignment-general-formulation-v2} is closely related to the problem of {\it graph isomorphism} defined as follows:
\begin{definition}[Graph Isomorphism]
Let $G_1=(V_1,E_1)$ and $G_2=(V_2,E_2)$ be two binary graphs. $G_1$ and $G_2$ are isomorphic if there exists a permutation matrix $P$ such that $G_1=P G_2 P^T$.
\end{definition}

The computational problem of determining whether two finite graphs are isomorphic is called the {\it graph isomorphism problem}. Moreover given two isomorphic graphs $G_1$ and $G_2$, in the graph isomorphism problem one aims to find the permutation matrix $P$ such that $G_1=P G_2 P^T$. The computational complexity of this problem is unknown \cite{complexity-isomorphisim-schweitzer}.

In the following lemma we formalize a connection between the graph alignment optimization and the classical graph isomorphism problem:
\begin{lemma}\label{lem:iso-align-relation1}
Let $G_1$ and $G_2$ be two isomorphic Erd\H{o}s-R\'enyi graphs \cite{erdHos1961strength} such that $Pr[G_1(i,j)=1]=p$ and $G_2=P G_1 P^T$, where $P$ is a permutation matrix. Let $p\neq 0,1$. Then, for any selection of scores $s_1>s_2>s_3>0$, $P$ maximizes the expected graph alignment objective function of Optimization \eqref{opt:alignment}. The expectation is over different realizations of $G_1$ and $G_2$.
\end{lemma}
\begin{proof}
The proof is presented in Section \ref{subsec:proof-lem:iso-align-relation2}.
\end{proof}

The result of Lemma \ref{lem:iso-align-relation1} can be extended to the case where edges of graphs are flipped through a random noise matrix:
\begin{lemma}\label{lem:iso-align-relation2}
Let $G_1$ be an Erd\H{o}s-R\'enyi graph such that $Pr[G_1(i,j)=1]=p$. Let $\tG_1$ be a graph resulting from flipping edges of $G_1$ independently and randomly with probability $q$. Suppose $G_2=P \tG_1 P^T$ where $P$ is a permutation matrix. Let $0<p<1/2$ and $0\leq q<1/2$. Then, for any selection of scores $s_1>s_2>s_3>0$, $P$ maximizes the expected graph alignment objective function of Optimization \eqref{eq:alignment-general-formulation-v2}. The expectation is over different realizations of $G_1$ and $G_2$.
\end{lemma}
\begin{proof}
The proof is presented in Section \ref{subsec:proof-lem:iso-align-relation2}.
\end{proof}

Finding an isomorphic mapping across sufficiently large Erd\H{o}s-R\'enyi graphs can be done efficiently with high probability (w.h.p.) through canonical labeling \cite{isom-czajka2008improved}. Canonical labeling of a graph consists of assigning a unique label to each vertex such that labels are invariant under isomorphism. The graph isomorphism problem can then be solved efficiently by mappings nodes with the same canonical labels to each other \cite{babai1979canonical}. One example of canonical labeling is the degree neighborhood of a vertex defined as a sorted list of neighborhood degrees of vertices \cite{isom-czajka2008improved}. Note that graph alignment formulation is more general than the one of graph isomorphism: graph alignment aims to find an optimal mappings across two graphs which are not necessarily isomorphic.

\section{EigenAlign Algorithm}\label{sec:eigenalign}
\subsection{Problem Formulation and Notation}\label{sec:problem-formulation}
Let $\by$ be a vectorized version of $X$. That is, $\by$ is a vector of length $n_1 n_2$ where, $y(i+(j'-1)n_1)=X(i,j')$. To simplify notation, define $y_{i,j'}\triangleq X(i,j')$. Two mappings $(i,j')$ and $(r,s')$ can be matches which cause overlaps, can be mismatches which cause errors, or can be neutrals (Figure \ref{fig:match-mismatch-illustration}-a).
\begin{definition}\label{def:match-mismatch}
Suppose $G_1=(V_1,E_1)$ and $G_2=(V_2,E_2)$ are undirected graphs. Let $\{i,r\}\subseteq V_1$ and $\{j',s'\}\subseteq V_2$ where $X(i,j')=1$ and $X(r,s')=1$. Then,
\begin{itemize}
\item[-] $(i,j')$ and $(r,s')$ are {\it matches} if $(i,r)\in E_1$ and $(j',s')\in E_2$.
\item[-] $(i,j')$ and $(r,s')$ are {\it mismatches} if only one of the edges $(i,r)$ and $(j',s')$ exists.
\item[-] $(i,j')$ and $(r,s')$ are {\it neutrals} if none of the edges $(i,r)$ and $(j',s')$ exists.
\end{itemize}
\end{definition}

Definition \ref{def:match-mismatch} can be extended to the case where $G_1$ and $G_2$ are directed graphs. In this case mappings $(i,j')$ and $(r,s')$ are matches/mismatches if they are matches/mismatches in one of the possible directions. However it is possible to have these mappings be matches in one direction while they are mismatches in the other direction (Figure \ref{fig:match-mismatch-illustration}-c). These mappings are denoted as {\it inconsistent mappings}, defined as follows:
\begin{definition}\label{def:inconsistent}
Let $G_1=(V_1,E_1)$ and $G_2=(V_2,E_2)$ be two directed graphs and $\{i,r\}\subseteq V_1$ and $\{j',s'\}\subseteq V_2$ where $X(i,j')=1$ and $X(r,s')=1$. If edges $i \to r$, $r \to i$, and $j' \to s'$ exist, however, $s' \to j'$ does not exist, then mappings $(i,j')$ and $(r,s')$ are {\it inconsistent}.
\end{definition}

Consider two undirected graphs $G_1=(V_1,E_1)$ and $G_2=(V_2,E_2)$. We form an {\it alignment graph} represented by adjacency matrix $A$ in which nodes are mapping edges across the original graphs, and the edges capture whether the pair of mapping edges are matches, mismatches or neutrals (Figure \ref{fig:framework}).

\begin{definition}\label{def:alignment-network}
Let $\{i,r\}\subseteq V_1$ and $\{j',s'\}\subseteq V_2$ where $X(i,j')=1$ and $X(r,s')=1$.
\begin{equation}\label{eq:alignment-net-expanded-form}
A\big[(i,j'),(r,s')\big]=\begin{cases}
    s_1,& \text{if } (i,j') \text{ and } (r,s') \text{ are matches},\\
    s_2,& \text{if } (i,j') \text{ and } (r,s') \text{ are neutrals},\\
    s_3, & \text{if } (i,j') \text{ and } (r,s') \text{ are mismatches},
\end{cases}
\end{equation}
where $s_1$, $s_2$, and $s_3$ are scores assigned to matches, neutrals, and mismatches, respectively. Without loss of generality we assume $s_1>s_2>s_3>0$.
\end{definition}

We can re-write \eqref{eq:alignment-net-expanded-form} as follows:
\begin{align}\label{eq:alignment-net}
 A\big[(i,j'),(r,s')\big]=&(s_1+s_2-2s_3)G_1(i,r)G_2(j',s')\\
 &+(s_3-s_2)(G_1(i,r)+G_2(j',s'))+s_2.\nonumber
\end{align}
We can summarize \eqref{eq:alignment-net-expanded-form} and \eqref{eq:alignment-net} as follows:
\begin{align}\label{eq:alignment-network-kron-product}
A=&(s_1+s_2-2s_3)(G_1\otimes G_2)+(s_3-s_2) (G_1\otimes \bbI_{n_2})\\
&+ (s_3-s_2) (\bbI_{n_1}\otimes G_2)+ s_2 (\bbI_{n_1} \otimes \bbI_{n_2}),\nonumber
\end{align}
where $\otimes$ represents matrix Kronecker product, and $\bbI_{n}$ is an $n\times n$ matrix whose elements are all ones.

A similar scoring scheme can be used for directed graphs. When graphs are directed, some mappings can be inconsistent according to Definition \ref{def:inconsistent}, i.e., they are matches in one direction and mismatches in another. Scores of inconsistent mappings can be assigned randomly to matched/mismatched scores, or to an average score of matches and mismatches (i.e., $(s_1+s_3)/2$). For random graphs, inconsistent mappings are rare events. For example, suppose graph edges are distributed according to a Bernoulli distribution with parameter $p$. Then, the probability of having an inconsistent mapping for a particular pair of paired nodes across graphs is equal to $4p^3 (1-p)$. Therefore, their effect in graph alignment is negligible, particularly for large sparse graphs. Throughout the paper, for directed graphs we assume inconsistent mappings have negligible effect, unless we mention the importance of such inconsistency explicitly.

In practice some mapping edges across two graphs may not be possible, owing to additional side information. The set of possible mapping edges across two graphs is denoted by $\cR=\{(i,j'):i\in V_1, j'\in V_2\}$. If $\cR=V_1\times V_2$, the problem of graph alignment is called {\it unrestricted}. If some mappings across two graphs are prevented (i.e., $X(i,j')=y_{i,j'}=0$, for $(i,j')\notin \cR$), then the problem of graph alignment is called {\it restricted}.

Using the vectorized version of $X$, the graph alignment optimization \eqref{eq:alignment-general-formulation} can be written as follows:
\begin{align}\label{opt:alignment}
\max_{\by} \quad&  \by^T A \by, \\
&\sum_{i}y_{i,j'}\leq 1, \quad \forall i\in V_1, \nonumber\\
&\sum_{j'}y_{i,j'}\leq 1, \quad \forall j'\in V_2, \nonumber\\
&y_{i,j'}\in\{0,1\}, \quad \forall (i,j')\in V_1\times V_2, \nonumber\\
&y_{i,j'}= 0, \quad\quad\quad\! \forall (i,j')\notin \cR, \nonumber
\end{align}
where $A$ is defined according to \eqref{eq:alignment-net} and $\cR\subseteq V_1\times V_2$ is the set of possible mapping edges across two graphs.

\subsection{EigenAlign Algorithm}\label{subsec:eigenalign-alg}
We now introduce {\it EigenAlign} (EA) algorithm which computes a solution for the graph alignment optimization \eqref{opt:alignment} leveraging spectral properties of graphs:
\begin{figure}[t]
  \centering
      \includegraphics[height=6cm]{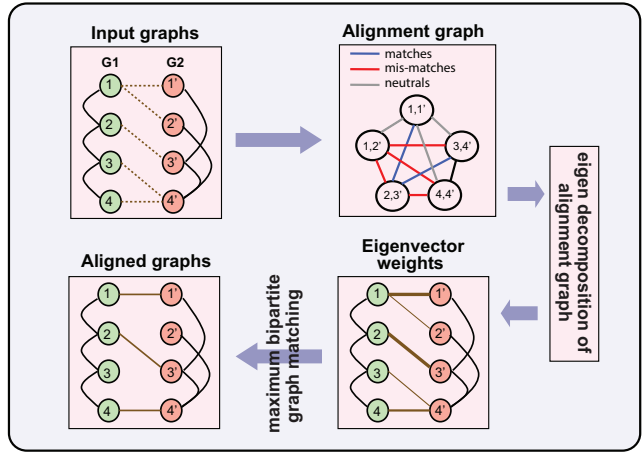}
  \caption{The Framework of EigenAlign algorithm \ref{alg:eigenalign}.}
  \label{fig:framework}
\end{figure}
\begin{alg}[EigenAlign Algorithm]\label{alg:eigenalign}
Let $G_1=(V_1,E_1)$ and $G_2=(V_2,E_2)$ be two binary graphs whose corresponding alignment graph is denoted by $A$ according to \eqref{eq:alignment-net}. EigenAlign algorithm solves the graph alignment optimization \eqref{opt:alignment} in two steps:

{\bf Step 1, An Eigenvector Computation Step:} In this step we compute $\bv$, an eigenvector of the alignment graph $A$ with the maximum eigenvalue.

{\bf Step 2, A Linear Assignment Step:} In this step we solve the following maximum weight bipartite matching optimization:
\begin{align}\label{opt:eigenalign}
\max_{\by} \quad & \bv^T \by, \\
&\sum_{j'}y_{i,j'}\leq 1, \quad \forall i\in V_1, \nonumber\\
&\sum_{i}y_{i,j'}\leq 1, \quad \forall j'\in V_2, \nonumber\\
&y_{i,j'}\in\{0,1\}, \quad \forall (i,j')\in  V_1\times  V_2, \nonumber\\
&y_{i,j'}=0, \quad\quad \forall (i,j')\notin \cR. \nonumber
\end{align}

\end{alg}
Algorithm \ref{alg:eigenalign} can be extended to directed graphs $G_1^{dir}$ and $G_2^{dir}$ as well. This framework is depicted in Figure \ref{fig:framework}. Below we provide intuition on different steps of the EigenAlign algorithm. For simplicity we assume all mappings across graphs are possible (i.e., $\cR=\{(i,j'):\forall i\in V_1, \forall j'\in V_2\}$). In the restricted graph alignment setup, without loss of generality, one can eliminate rows and columns of the alignment matrix corresponding to mappings that are not allowed.

In the eigen decomposition step of EigenAlign, we ignore bijective constraints (i.e., constraints $\sum_{i}y_{i,j'}\leq 1$ and $\sum_{j'}y_{i,j'}\leq 1$) because they will be satisfied in the second step of the algorithm through a linear optimization. Under these assumptions Optimization \eqref{opt:alignment} can be simplified to the following optimization:
\begin{align}\label{opt:relax1}
\quad\quad\quad\quad\quad\quad\quad\quad\quad\quad\max_{\by}  \quad & \by^T A \by, \\
&y_{i,j'}\in\{0,1\}, \quad \forall (i,j')\in V_1\times V_2. \nonumber
\end{align}
To approximate a solution of this optimization, we replace integer constraints with constraints over a hyper-sphere restricted by hyper-planes (i.e., $\|\by\|_{2}\leq 1$ and $\by\geq 0$). Thus, optimization \eqref{opt:relax1} is simplified to the following:
\begin{align}\label{opt:relax2}
\max_{\by} \quad & \by^T A \by, \\
&\|\by\|_{2} \leq 1, \nonumber\\
& \by\geq 0. \nonumber
\end{align}
In the following, we show that $\bv$, the leading eigenvector of the alignment matrix $A$, is an optimal solution of Optimization \eqref{opt:relax2}. Suppose $\by_1$ is an optimal solution of Optimization \eqref{opt:relax2}. Let $\by_2$ be a solution of the following optimization without non-negativity constraints:
\begin{align}\label{opt:relax3}
\max_{\by} \quad & \by^T A \by, \\
&\|\by\|_2 \leq 1. \nonumber
\end{align}
Following the Rayleigh-Ritz formula \cite{trefethen1997numerical}, the leading eigenvector of the alignment matrix is an optimal solution of Optimization \eqref{opt:relax3} (i.e., $\by_2=\bv$). Now we use the following theorem to show that in fact $\by_1=\bv$:
\begin{theorem}\label{thm:perron}
Suppose $A$ is a matrix whose elements are strictly positive. Let $\bv$ be an eigenvector of $A$ corresponding to the largest eigenvalue. Then, $\forall i$, $v_i>0$. Moreover, all other eigenvectors must have at least one negative, or non-real component.
\end{theorem}
\begin{proof}
See e.g., reference \cite{pillai2005perron} (Theorem 1).
\end{proof}

Since $\by_2$ is a solution of Optimization \eqref{opt:relax3}, we have $\by_2^T A \by_2 \geq \by_1^T A \by_1$. Using this inequality along with the Perron-Frobenius Theorem lead to $\by_1=\bv$, as the unique solution of optimization \eqref{opt:relax2}.

The solution of the eigen decomposition step assigns weights to all possible mapping edges across graphs ignoring bijective constraints (constraints $\sum_{j'}y_{i,j'}\leq 1$ and $\sum_{i}y_{i,j'}\leq 1$). However, in the graph alignment setup, each node in one graph can be mapped to at most one node in the other graph. To satisfy these constraints, we use eigenvector weights in a linear optimization framework of maximum weight bipartite matching setup of Optimization \eqref{opt:eigenalign} \cite{book-graphtheory-west2001introduction}.

\subsection{Computational Complexity of EigenAlign}
Let the number of nodes of graphs $G_1$ and $G_2$ be $\cO(n)$. Let $k=|\cR|$ be the number of possible mappings across two graphs. In an unrestricted graph alignment setup, we have $k=\cO(n^2)$. However, in a restricted graph alignment, $k$ may be significantly smaller than $n^2$. EigenAlign has three steps:

(i) Forming an alignment graph $A$ that has a computational complexity of $\cO(k^2)$, as all pairs of possible mappings should be considered.

(ii) An eigen decomposition step where we compute the leading eigenvector of the alignment graph. This operation can be performed in $\cO(k^2)$ computational complexity using QR algorithms and/or power methods \cite{complexity-eigen}. Therefore, the worst case computational complexity of this part is $\cO(k^2)$.

(iii) A maximum weight bipartite matching algorithm step, that can be solved efficiently using linear programming or the Hungarian algorithm \cite{book-graphtheory-west2001introduction}. The worst case computational complexity of this step is $\cO(n^3)$. If the set $\cR$ has a specific structure (e.g., small subsets of nodes in one graph are allowed to be mapped to small subsets of nodes in the other graph), this cost can be reduced significantly.
\begin{proposition}\label{prop:complexity}
The worst case computational complexity of the EigenAlign Algorithm is $\cO(k^2+n^3)$.
\end{proposition}
\begin{remark}\label{remark:greedy}
\textup{For large graphs, to reduce the overall computational complexity, the linear assignment optimization may be replaced by a greedy bipartite matching algorithm (e.g., \cite{preis1999linear}). In the greedy matching approach, at each step, the heaviest possible mapping is added to the current matching until no further mappings can be added. It is straightforward to show that this greedy algorithm finds a bipartite matching whose weight is at least half the optimum. The computational complexity of this greedy algorithm is $\cO(k\log(k)+nk)$.}
\end{remark}

If we only consider matches in the graph alignment optimization (i.e., $s_2=s_3=0$ in \eqref{eq:alignment-general-formulation}), the complexity of the eigen decomposition step can be reduced, since we need to compute top eigenvectors of sparse adjacency matrices. By considering mismatches, eigenvector computation should be performed over dense matrices, which require a higher computational complexity.

\subsection{Mean-field Optimality of EigenAlign Over Erd\H{o}s-R\'enyi Graphs}\label{subsec:erdos}
Here we analyze the performance of the EigenAlign algorithm over Erd\H{o}s-R\'enyi graphs, for both isomorphic and non-isomorphic cases, under two different noise models. While real graphs often have different structures than Erd\H{o}s-R\'enyi graphs, we consider this family of graphs in this section owing to their analytical tractability.

Suppose $G_1=(V_1,E_1)$ is an undirected Erd\H{o}s-R\'enyi graph with $n$ nodes where $Pr[G_1(i,j)=1]=p$ for $1\leq i,j\leq n$. Suppose $\tG$ is a noisy version of the graph $G_1$. We consider two different noise models in this section:

{\bf Noise Model I:} In this model we have,
\begin{align}\label{eq:noise-model1}
\tG_1\triangleq G_1 \odot (\mathbb{1}-Q)+ (\mathbb{1}-G_1) \odot Q,
\end{align}
where $\odot$ represents the Hadamard product, $\mathbb{1}$ is the matrix of all ones, and $Q$ is a binary symmetric random matrix whose edges are drawn i.i.d. from a Bernoulli distribution with $Pr[Q(i,j)=1]=p_e$. In words, the operation $G_1 \odot (1-Q)+ (1-G_1) \odot Q$ flips edges of $G_1$ uniformly randomly with probability $p_e$.

{\bf Noise Model II:} In this model we have,
\begin{align}\label{eq:noise-model2}
\tG_1\triangleq G_1 \odot (1-Q)+ (1-G_1) \odot Q',
\end{align}
where $Q$ and $Q'$ are binary symmetric random matrices whose edges are drawn i.i.d. from a Bernoulli distribution with $Pr[Q(i,j)=1]=p_e$ and $Pr[Q'(i,j)=1]=\pet$. Under this model, edges of $G_1$ flip uniformly randomly with probability $p_e$, while non-connecting tuples in $G_1$ will be connected in $\tG_1$ with probability $\pet$. Because $G_1$ is an Erd\H{o}s-R\'enyi graph with parameter $p$, choosing
\begin{align}\label{eq:def-q2}
\pet=\frac{p p_e}{1-p},
\end{align}
leads to having the expected density of graphs $G_1$ and $G_2$ be equal to $p$.

Using either model I \eqref{eq:noise-model1} or model II \eqref{eq:noise-model2} for $\tG_1$, we define $G_2$ as follows:
\begin{align}\label{eq:define-g2-erdos}
G_2 \triangleq P \tG_1 P^T,
\end{align}
where $P$ is a permutation matrix. Recall that $\cR$ is the set of possible mapping edges across graphs $G_1$ and $G_2$. Throughout this section, we assume that we are in the restricted graph alignment regime where $|\cR|=kn$ for $k>1$. The $n$ true mapping edges ($i\leftrightarrow i'$ if $P=I$) are included in $\cR$, while the remaining $(k-1)n$ mappings are selected uniformly randomly.

Let $S_{\text{true}}$ be the set of true mapping edges between $G_1$ and $G_2$, i.e., $S_{\text{true}}\triangleq \{(i,j): P(i,j)=1\}$. We define $S_{\text{false}}=\cR-S_{\text{true}}$ as the set of incorrect mapping edges between the two graphs. Moreover, we choose scores assigned to matches, neutrals and mismatches as $s_1=\alpha+\eps$, $s_2=1+\eps$ and $s_3=\eps$, respectively, where $\alpha>1$ and $0<\eps\ll 1$. These selections satisfy score conditions $s_1>s_2>s_3>0$ and lead to the regularization parameter $\gamma=1/(1+\alpha)$ in \eqref{eq:alignment-general-formulation-v2}.
\begin{theorem}\label{thm:erdos-noisy}
Let $A$ be the alignment graph between $G_1$ and $G_2$ as defined in \eqref{eq:alignment-net} with $s_1=\alpha+\eps$, $s_2=1+\eps$ and $s_3=\eps$. Let $\bv$ be the eigenvector of $\EE[A]$ corresponding to the largest eigenvalue, where the expectation is over realizations of $G_1$, $G_2$ and $\cR$. Then, under both noise models \eqref{eq:noise-model1} and \eqref{eq:noise-model2}, if $0<p<1/2$, and $0 \leq p_e<1/2$, as $n\to\infty$,
\begin{equation}
\bv(t_1)>\bv(t_2),\quad\forall t_1\in S_{\text{true}}~\text{and}~ \forall t_2\in S_{\text{false}}\nonumber.
\end{equation}
\end{theorem}

In noise models \eqref{eq:noise-model1} and \eqref{eq:noise-model2}, if we put $p_e=0$, then $G_2$ is isomorphic with $G_1$ because there exists a permutation matrix $P$ such that $G_2=P G_1 P^T$. For this case, we have the following Corollary:
\begin{corollary}\label{thm:erdos}
Let $G_1$ and $G_2$ be two isomorphic Erd\H{o}s-R\'enyi graphs with $n$ nodes such that $G_1=P G_2 P^T$, where $P$ is a permutation matrix. Under the conditions of Theorem \ref{thm:erdos-noisy}, as $n\to\infty$, $\bv(t_1)>\bv(t_2)$ where where $\bv$ is the top eigenvector of the expected alignment graph, $t_1$ is a true mapping edge and $t_2$ is a false mapping edge between the two graphs.
\end{corollary}

We present proofs of Theorem \ref{thm:erdos-noisy} and Corollary \ref{thm:erdos} in Sections \ref{subsec:proof-thm-erdos} and \ref{subsec:proof-erdos-noisy}.

In the EigenAlign algorithm, we use values of the top eigenvector of the alignment graph in a maximum weight bipartite matching optimization to extract bijective mappings between the two graphs. Thus, if true mapping edges obtained higher eigenvector scores compared to the false one, the EigenAlign algorithm would infer optimal mappings between the two graphs. Theorem \ref{thm:erdos-noisy} indicates that, in an expectation sense, true mapping edges obtain larger eigenvector scores compared to the false ones when $|\cR|=kn$. In Section \ref{sec:eval} and through simulations, we show that the error of the EigenAlign algorithm is empirically small even in an unrestricted graph alignment setup.


\section{LowRankAlign Algorithm}\label{sec:lowrank}
In this section, we introduce a graph alignment algorithm that uses higher-order eigenvectors of (transformations of) adjacency graphs to align their structures. We refer to this extension as {\it LowRankAlign} (LRA). LRA can be useful specially in cases where leading eigenvectors of graphs are not informative. This case occurs for instance in the alignment of regular graph structures. Moreover, LRA does not require an explicit formation of the alignment graph which can be costly for large graphs if all mappings across graphs are possible.

Higher order eigenvectors have been used in other spectral inference problems such as graph clustering \cite{newman2006modularity,sussman,qin,athreya,saade} and the matrix coupling \cite{fraikin2008gradient,fraikin2008optimizing}. Moreover reference \cite{knossow2009inexact} has used higher order eigenvectors of the graph Laplacian to embed large graphs on a low-dimensional isometric space to compute an inexact matching. Our goal in this section is to provide a principled framework to exploit higher order eigenvectors in the graph alignment problem.

Here we assume graphs are symmetric. For simplicity we assume $n_1=n_2=n$. All discussions can be extended to the case where $n_1\neq n_2$.  Moreover, to simplify analysis, we assume singular values of matrices have multiplicity of one. Let $\sqcap$ be the set of all permutation matrices of size $n\times n$. Thus, the graph alignment optimization can be written as follows \footnote{To consider the generalized graph alignment formulation of \eqref{eq:alignment-compact-formulation}, one can replace $G_1$ and $G_2$ with $G_1-\gamma\bbI$ and $G_2-\gamma\bbI$ in \eqref{eq:qap-overlap-perm}, respectively.}:
\begin{align}\label{eq:qap-overlap-perm}
\max &\quad Tr (G_1 X G_2 X^T),\\
& X \in \sqcap. \nonumber
\end{align}
Let $X^*$ be an optimal solution of optimization \eqref{eq:qap-overlap-perm}. Finding an optimal solution of this optimization is known to be NP-hard \cite{QAP-newlp-NPhard}. If $X\in\sqcap$, we have
\begin{align}
Tr (G_1 X G_2 X^T)= Tr \big((G_1+\delta_1 I) X (G_2+\delta_2 I) X^T\big)+\text{constant}.
\end{align}
In other words we can add and subtract multiples of identity to make the resulting symmetric matrices positive definite, without changing the structure of the problem. Thus, without loss of generality, we assume that matrices $G_1$ and $G_2$ are positive semi-definite.

We compute a solution for Optimization \eqref{eq:qap-overlap-perm} in two steps:

(i) {\bfseries The Relaxation Step:} First, we compute a solution $X_0$ to a relaxation of Optimization \eqref{eq:qap-overlap-perm} over orthogonal matrices. Other relaxations can be considered as well. $X_0$ may not be a valid permutation matrix.

(ii) {\bfseries The Rounding Step:} We propose a rounding step using projection in the direction of eigenvectors of (transformations of) adjacency graphs scaled by their corresponding eigenvalues.

Below we explain these steps with more details:

{\bfseries The Relaxation step:} Let $\Gamma$ be a set that contains all permutation matrices (i.e., $\sqcap\subseteq\Gamma$). An example of $\Gamma$ is the set of orthogonal matrices. Let $X_{0}$ be a solution of the following optimization:
\begin{align}\label{eq:qap-overlap-relaxation}
\max &\quad Tr (G_1 X G_2 X^T),\\
& X \in \Gamma. \nonumber
\end{align}

If $\Gamma$ is assumed to be the set of orthogonal matrices (i.e., $\Gamma=\cO$), an optimal solution of optimization \eqref{eq:qap-overlap-relaxation} can be found using eigen decomposition of matrices $G_1$ and $G_2$ as follows:
\begin{theorem}\label{thm:orth-optimal}
Suppose $v_i$ and $u_i$ are eigenvectors of symmetric matrices $G_1$ and $G_2$, respectively. Let $V$ and $U$ be eigenvector matrices whose $i$-th columns are $v_i$ and $u_i$, respectively. Then,
\begin{align}
X_0=VU^{T}=\sum_{i=1}^{n} v_i u_i^T,
\end{align}
is an optimal solution of optimization \eqref{eq:qap-overlap-relaxation} over orthogonal matrices (i.e., $\Gamma=\cO$).
\end{theorem}
\begin{proof}
See Section 6.1 of reference \cite{orthogonal-finke1987quadratic}.
\end{proof}

Theorem \ref{thm:orth-optimal} characterizes an optimal solution of the orthogonal relaxation of the graph alignment optimization. A similar argument can be constructed for eigenvectors of the matrix $G_2$.  Let
\begin{align}\label{eq:orth-relax-all-solutions}
\cX_0\triangleq \big\{X_0: X_0=\sum_{i=1}^{n} s_i v_i u_i^T, \bs\in \{-1,1\}^{n}\big\},
\end{align}
where $s_i$ is the $i$-th component of the vector $\bs$. The set $\cX_0$ represents multiple optimal solutions of optimization \eqref{eq:qap-overlap-relaxation} when $\Gamma=\cO$. It is because if $\bv$ is an eigenvector of a matrix corresponding to the eigenvalue $\lambda$, $-\bv$ is also an eigenvector of the same matrix with the same eigenvalue. $\cX_0$ can have at most $2^n$ distinct members.

{\bfseries The Rounding step:} $X_0$ may not be a valid permutation matrix. One way to find a permutation matrix using $X_0$ is to project $X_0$ over the space of permutation matrices $\sqcap$:
\begin{align}\label{eq:proj}
\max &\quad Tr (X X_0^T),\\
& X \in \sqcap.\nonumber
\end{align}
However, it has been shown that an optimal solution of optimization \eqref{eq:proj} has a poor performance in practice \cite{orth-qap-alg}. In the following, we propose an alternative algorithm to compute a permutation matrix using $X_0$ with a certain performance guarantee. Consider the following optimization:
\begin{align}\label{eq:qap-overlap-linear}
\max &\quad Tr (G_1 X_0 G_2 X^T),\\
& X \in \sqcap, \nonumber\\
& X_0\in \cX_0.\nonumber
\end{align}
For a fixed $X_0$, this is a maximum weight bipartite matching optimization which can be solved exactly using linear programming. Let $X_{lin}^*$ be an optimal solution of optimization \eqref{eq:qap-overlap-linear}. Define
\begin{align}\label{eq:def-f-fbar}
f(X) & \triangleq Tr (G_1 X G_2 X^T),\\
\tf(X) & \triangleq Tr (G_1 X_0 G_2 X_0^T)+2Tr (G_1 X_0 G_2 (X-X_0)^T).\nonumber
\end{align}
\begin{theorem}\label{thm:gaurantee}
Let $X^*$ and $X_{lin}^*$ be optimal solutions of optimizations \eqref{eq:qap-overlap-perm} and \eqref{eq:qap-overlap-linear}, respectively. We have,
\begin{align}\label{eq:gaurantee}
|f(X^*)-\tf(X_{lin}^*)|\leq \epsilon^2 \sum_{i=1}^n \sigma_{i}(G_1) \sigma_{i}(G_2),
\end{align}
where $\sigma_{i}(G_a)$ represents the $i$-th largest singular value of matrix $G_a$, for $a=1,2$, and $\epsilon$ is a bound on the relaxation gap (i.e., $\min_{X_0\in\cX_0}~~\|X^*-X_0\|_{op}\leq \epsilon$). Note that $\|.\|_{op}$ indicates the matrix operator norm.
\end{theorem}
\begin{proof}
See Section \ref{subsec:proof-thm-gaurantee}.
\end{proof}

Optimization \eqref{eq:qap-overlap-linear} can be simplified to the following optimization which finds a valid permutation matrix using the orthogonal relaxation of the graph alignment optimization:
\begin{align}\label{eq:qap-overlap-linear-orth}
\max &\quad Tr \big((\sum_{i=1}^{n} \lambda_{i}(G_1)\lambda_{i}(G_2) s_i v_i u_i^T) X^T\big),\\
& X \in \sqcap, \nonumber\\
& \bs\in \{-1,1\}^{n},\nonumber
\end{align}
where $\lambda_i(G_a)$ is the $i$-th largest eigenvalue of $G_a$ for $a=1,2$. The objective function of optimization \eqref{eq:qap-overlap-linear-orth} simplifies the graph alignment problem to the simultaneous alignment of eigenvectors whose contributions in the overall alignment score are weighed by their corresponding eigenvalues. However, there are possibly exponentially many optimal solutions for optimization \eqref{eq:qap-overlap-linear-orth} and obtaining their resulting permutation matrices would be computationally infeasible. Because contributions of eigenvectors with small eigenvalues to the objective function of optimization \eqref{eq:qap-overlap-linear-orth} are small, one can instead as a heuristic, presumably solve the following optimization based on the low rank approximation of the objective function:
\begin{alg}[LowRankAlign Algorithm]\label{alg:LRA}
The following optimization summarizes the LRA algorithm:
\begin{align}\label{eq:qap-overlap-linear-orth-low-rank}
\max \quad & Tr \big((\sum_{i=1}^{k} s_i \lambda_{i}(G_1) \lambda_{i}(G_2) v_i u_i^T) X^T\big),\\
& X \in \sqcap, \nonumber\\
& s_i \in \{-1,1\}, \quad \forall 1\leq i\leq k. \nonumber
\end{align}
\noindent
where $k$ is a constant that determines the rank of the affinity matrix.
\end{alg}

In the restricted graph alignment setup, some mapping edges across two graphs may not be allowed. In that case, one can set the affinity weights (i.e., weights used in the maximum weight bipartite matching step) of such pairs in optimization \eqref{eq:qap-overlap-linear-orth-low-rank} to be $-\infty$.

\section{Performance Evaluation Over Synthetic Graphs}\label{sec:eval}

\begin{figure*}[t]
  \centering
      \includegraphics[width=14cm]{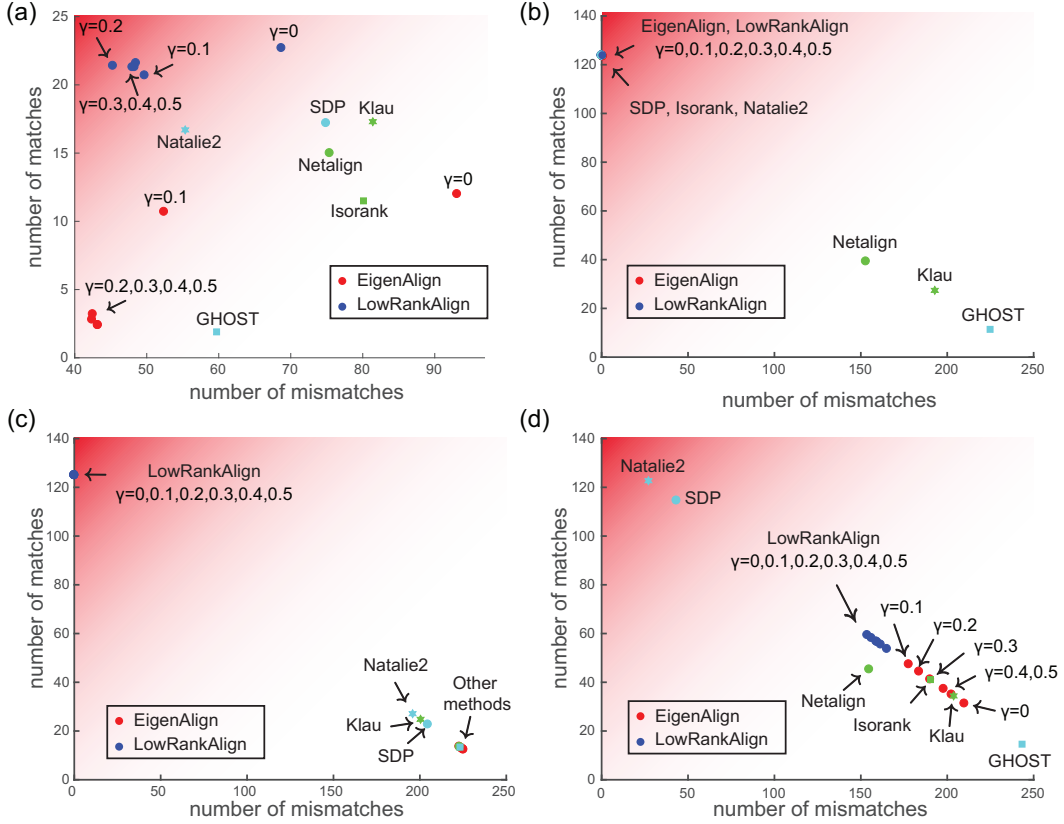}
  \caption{Performance evaluation of different graph alignment methods over (a) stochastic block models, (b) isomorphic Erd\H{o}s-R\'enyi graphs, (c) isomorphic random regular graphs, and (d) noisy power law graphs. Experiments have been repeated 10 times in each case. For each method the average number of matches and mismatches have been shown. The high-match low-mismatch area has been highlighted by red shades.}
  \label{fig:sim-res-all}
\end{figure*}

Here we compare the performance of the proposed graph alignment algorithms (LRA and EA) against some other graph alignment methods including Natalie2 \cite{natalie2,natalie}, GHOST \cite{patro2012global}, IsoRank \cite{isorank}, NetAlign \cite{bayati2013message}, Klau's approach \cite{klau2009new} as well as an SDP-based method \cite{splitting-sdp} through simulations. Natalie2 and Klau's approach use Lagrange multipliers to relax the underlying quadratic assignment problem. IsoRank is a global graph alignment method that uses an iterative approach to align nodes across two graphs based on their neighborhood similarities, while GHOST uses a heuristic multi-scale spectral signature of graphs to compute an alignment across them. NetAlign formulates the alignment problem in a quadratic optimization framework and uses message passing to approximately solve it.  The SDP-based method \cite{splitting-sdp} uses a convex relaxation of the underlying QAP based on matrix splitting. In our simulations, we use default parameters of these methods.

We report the performance of proposed EigenAlign (EA) and LowRankAlign (LRA) Algorithms for $\gamma\in\{0,0.1,0.2,0.3,0.4,0.5-\epsilon\}$ where $\eps=0.001$. In general, this parameter can be tuned in different applications using standard machine learning techniques such as cross validations \cite{cross-validation}. For LRA we use top $k=3$ eigenvectors of input graphs as larger values of $k$ did not have a significant effect on the results. We consider four different setups:
\begin{itemize}
  \item[-] $G_1$ is an Erd\H{o}s-R\'enyi graph with $n_1=25$ nodes and the density parameter $0.1$. $G_2$ is a stochastic block model with two blocks each with 25 nodes (i.e., $n_2=50$). Edge densities within blocks are 0.1 and 0.3, and the edge density across blocks is 0.05.
  \item[-] $G_1$ and $G_2$ are isomorphic Erd\H{o}s-R\'enyi graphs with $n_1=n_2=50$ nodes with an edge density $0.1$.
  \item[-] $G_1$ and $G_2$ are isomorphic random regular graphs with $n_1=n_2=50$ nodes whose edge density parameters are $0.1$.
  \item[-] $G_1$ is a power law graph \cite{power-law2001random} constructed as follows: we start with a random subgraph with 5 nodes. At each iteration, a node is added to the graph connecting to three existing nodes with probabilities proportional to their degrees. This process is repeated till the number of nodes in the graph is equal to $n_1=50$. Then we construct $G_2$ according to the noise model \eqref{eq:noise-model2} with $p_e=0.05$. In \eqref{eq:noise-model2} we use the density of $G_1$ as parameter $p$.
\end{itemize}

Figure \ref{fig:sim-res-all} shows the number of matches and mismatches caused by different graph alignment methods in four considered setups for an unrestricted graph alignment problem. In the stochastic block model case (panel a) LRA outperforms other methods in terms of resulting in large number of matches and few mismatches. Since LRA with $\gamma=0$ ignores the effect of mismatches, it results in a slightly larger number of matches compared to the case with $\gamma\neq 0$. At the same time LRA with $\gamma=0$ results in a larger number of mismatches compared to the case with $\gamma\neq 0$. This highlights the effect of considering mismatches in the generalized graph alignment formulation \eqref{eq:alignment-general-formulation-v2} when graphs have different sizes and heterogenous edge densities.

Over isomorphic Erd\H{o}s-R\'enyi graphs (panel b), EA, LRA, Isorank, SDP and Natalie2 have the best performance of achieving the highest number of matches and zero mismatches. Netalign, Klau and GHOST have poor performances in this case. Note that some of these methods are designed for very sparse graphs and for the restricted graph alignment setup. This may partially explain the poor performance of these methods.

Over isomorphic random graphs (panel c) LRA outperforms other methods achieving the highest number of matches and zero mismatches. The performance of LRA is also robust against parameter $\gamma$. Note that the alignment of regular graph structures is one of the most difficult graph alignment cases because of homogeneity of node degrees. The fact that LRA performs well in this case while all other methods have poor performance illustrates the effectiveness of using higher order eigenvectors in aligning homogenous graph structures. Finally, over noisy power law graphs (panel d) Natalie2 and SDP outperform other methods. The performance of LRA in this case is higher than other methods except Natalie2 and SDP.

\begin{figure*}[t]
  \centering
      \includegraphics[width=14cm]{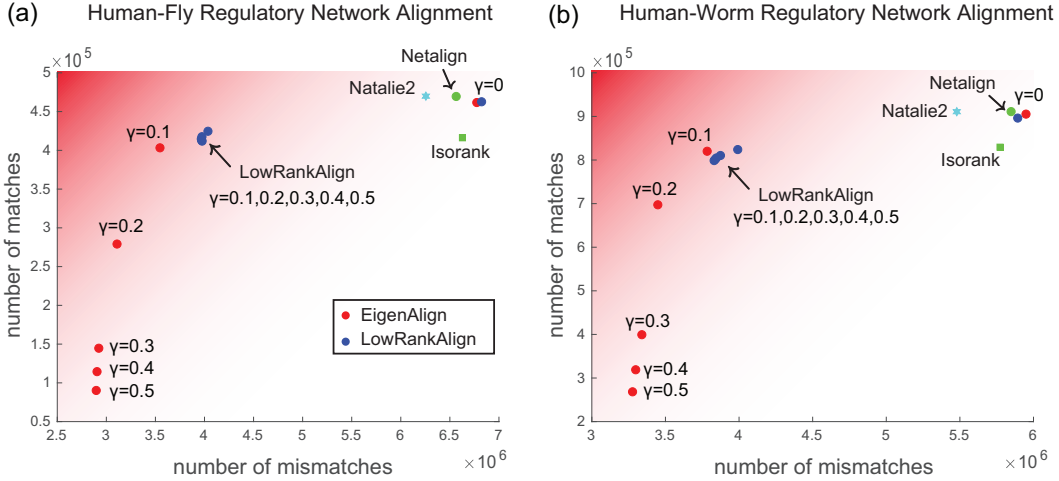}
  \caption{Performance evaluation of different graph alignment methods over (a) human-fly and (b) human-worm gene regulatory graphs. The high-match low-mismatch area has been highlighted by red shades.}
  \label{fig:real-data}
\end{figure*}

\section{Performance Evaluation Over Gene Networks}\label{sec:regulatory}
Here we apply graph alignment methods to compare gene regulatory graphs across human, fly and worm species. Comparative graph analysis in evolutionary studies often requires having a one-to-one mapping across genes of two or multiple networks. However, since human, fly and worm are distant species and as a result, many gene families have undergone extensive duplications and losses, we observe non-bijective homolog mappings across their genes \cite{jessica-paper}. For example, one gene in human can be homologous to multiple genes in fly and vice versa. To infer bijective mappings as a subset of homolog genes across species, we use graph alignment methods. We use regulatory networks that are inferred by integrating genome-wide functional and physical genomics datasets from ENCODE and modENCODE consortia (see the Appendix for more details).

Similarly to our discussion in Section \ref{sec:eval} we report the performance of proposed EigenAlign (EA) and LowRankAlign (LRA) methods for $\gamma\in\{0,0.1,0.2,0.3,0.4,0.5-\epsilon\}$ where $\eps=0.001$. For LRA we use top $k=2$ eigenvectors of input graphs. We also assess the performance of NetAlign, IsoRank, and Natalie2 in our real data analysis. We exclude Klau's approach \cite{klau2009new} and the SDP-based method of \cite{splitting-sdp} from our analysis in this section owing to their high memory and computational complexity. Moreover the GHOST method failed to run over these networks owing to some implementation errors.

Figure \ref{fig:real-data} shows the number of matches and mismatches caused by different graph alignment methods across human-fly and human-worm networks. In both cases EA and LRA with $\gamma=0$ (i.e., ignoring mismatches) have a comparable performance to other methods. However, by changing $\gamma$ we observe a trade-off between number of caused matches and mismatches. For example, in the human-fly network alignment case LRA with a non-zero $\gamma$ results in approximately 2-fold decrease in the number of mismatches while the number of caused matches decreases by approximately $10\%$. This highlights the effect of considering mismatches in the graph alignment optimization. To substantiate these inferences, further experiments should be performed to determine the involvement of inferred conserved gene interactions in different biological processes, which is beyond the scope of the present paper.

\section{Conclusion}\label{sec:conc}
In this paper, we made two main contributions to the field of graph alignment. Firstly, we proposed a generalized graph alignment formulation that considers both matches and mismatches in a standard QAP formulation. We showed that this can be critical in applications where graphs have different sizes and heterogenous edge densities. Secondly, we proposed two graph alignment algorithms which employ spectral decompositions of functions of adjacency graphs followed by a maximum weight bipartite matching optimization. One of our proposed methods simplifies the graph alignment optimization to simultaneous alignment of eigenvectors of (transformations of) adjacency graphs scaled by corresponding eigenvalues. We demonstrated effectiveness of the proposed methods theoretically for certain classes of graphs and over various synthetic and real graph models.

\section{Code}
We provide code for the proposed method in the following link:
\textcolor{blue}{https://github.com/SoheilFeizi/spectral-graph-alignment}

\section{Proofs}\label{sec:proofs}
In this section, we present proofs of the main results of the paper.
\subsection{Proofs of Lemmas \ref{lem:iso-align-relation1} and \ref{lem:iso-align-relation2}}\label{subsec:proof-lem:iso-align-relation2}
First we prove Lemma \ref{lem:iso-align-relation1}. Let $A$ be the alignment graph of $G_1$ and $G_2$. By a slight abuse of notation, we use $A$ as the adjacency matrix of the alignment graph as well. Suppose $\tP$ is a permutation matrix where $\rho\triangleq \frac{1}{2n}\|P-\tP\|>0$. Let $\by$ and $\bty$ be vectorized versions of permutation matrices $P$ and $\tP$, respectively. Then, we have,
\begin{align}
&\frac{1}{n^2} \EE[\bty_2^T A \bty_2] = (1-\rho) \big[p s_1+ (1-p) s_2\big]\\
&+ \rho \big[p^2 s_1+ (1-p)^2 s_2+ 2p(1-p)s_3\big]\nonumber\\
&< (1-\rho) \big[p s_1+ (1-p) s_2\big]+ \rho \big[p^2 s_1+ (1-p)^2 s_2+ p(1-p) (s_1+s_2)\big]\nonumber\\
&= (1-\rho) \big[p s_1+ (1-p) s_2\big]+ \rho \big[p s_1+ (1-p) s_2\big]\nonumber\\
&= p s_1+ (1-p) s_2\nonumber\\
&= \frac{1}{n^2} \EE[\by_1^T A \by_1].\nonumber
\end{align}

Now we prove Lemma \ref{lem:iso-align-relation2}: Similarly to the proof of Lemma \ref{lem:iso-align-relation1}, let $A$ be the alignment graph of $G_1$ and $G_2$ and suppose $\tP$ is a permutation matrix where $\rho\triangleq \frac{1}{2n}\|P-\tP\|>0$. Let $\by$ and $\bty$ be vectorized versions of permutation matrices $P$ and $\tP$, respectively. Define $a'$ and $b'$ as follows:
\begin{align}\label{eq:a-b-noisy-gen}
a' \triangleq& p(1-q)s_1+(1-p)(1-q) s_2+ (pq+(1-p)q)s_3, \\
b' \triangleq& \big(p^2(1-q)+p q(1-p)\big)s_1\nonumber\\
+& \big((1-p)^2(1-q)+p q (1-p)\big) s_2\nonumber\\
+& \big( 2 p (1-p) (1-q)+ 2 p^2 q \big) s_3.\nonumber
\end{align}
Thus,
\begin{align}\label{eq:a'-b'-diff}
a'-b'=p(1-p)(1-2q)(s_1+s_2-2s_3)+q(1-2p)s_3.
\end{align}
Because $s_1>s_2>s_3$, we have, $s_1+s_2-2s_3>0$. Because $0<p<1/2$ and $0\leq q<1/2$, we have $(1-2p)>0$ and $(1-2q)>0$. Therefore, according to \eqref{eq:a'-b'-diff}, $a'>b'$. Thus we have,
\begin{align}
\frac{1}{n^2} \EE[\bty^T A \bty] = (1-\rho) a'+ \rho b' < a'= \frac{1}{n^2} \EE[\by^T A \by].\nonumber
\end{align}
\subsection{Proof Of Corollary \ref{thm:erdos}}\label{subsec:proof-thm-erdos}
Without loss of generality and to simplify notations, we assume the permutation matrix $P$ is equal to the identity matrix $I$, i.e., the isomorphic mapping across $G_1$ and $G_2$ is $\{1\leftrightarrow 1', 2\leftrightarrow 2', \ldots, n\leftrightarrow n'\}$ (otherwise, one can relabel nodes in either $G_1$ or $G_2$ to have $P$ equal to the identity matrix). Therefore, $G_1(i,j)=G_2(i',j')$ for all $1\leq i,j\leq n$. Recall that $\by$ is a vector of length $kn$ which has weights for all possible mapping edges $(i,j')\in\cR$. To simplify notations and without loss of generality, we re-order indices of vector $\by$ as follows:
\begin{itemize}
\item[-] The first $n$ indices of $\by$ correspond to correct mappings, i.e., $y(1)=y_{1,1'}, y(2)=y_{2,2'}, \ldots, y(n)=y_{n,n'}$.
\item[-] The remaining $(k-1)n$ indices of $\by$ correspond to incorrect mappings. e.g., $y(n+1)=y_{1,2'}, y(n+2)=y_{1,3'}, \ldots, y(kn)=y_{r,s'}$ ($r\neq s$).
\end{itemize}
Therefore, we can write,
$$
\by=\begin{lbmatrix}{1}
  \by_1 \\
  \hline
  \by_2
\end{lbmatrix},
$$
where $\by_1$ and $\by_2$ are vectors of length $n$ and $(k-1)n$, respectively.

We re-order rows and columns of the alignment matrix $A$ accordingly. Define the following notations: $\cS_1=\{1,2,\ldots,n\}$ and $\cS_2=\{n+1,n+2,\ldots,kn\}$. The alignment matrix $A$ for graphs $G_1$ and $G_2$ can be characterized using equation \eqref{eq:alignment-net} as follows:
\begin{align}\label{eq:alignment-erdos}
&A(t_1,t_2)=\\
&\begin{cases}
    (\alpha+1)G_1(i,j)G_2(i',j')-G_1(i,j)-G_2(i',j')+1+\eps, \\
    \text{if } t_1\sim(i,i'),t_2\sim(j,j'), t_1\ \text{and}\ t_2\in \cS_1, t_1\neq t_2. \\
    (\alpha+1)G_1(i,j)G_2(r',s')-G_1(i,j)-G_2(r',s')+1+\eps\\
    \text{if } t_1\sim(i,r'),t_2\sim(j,s'),t_1\ \text{or}\ t_2\in \cS_2, t_1\neq t_2.\\
    1+\eps,\\
    \text{if } t_1=t_2,
\end{cases}\nonumber
\end{align}
where notation $t_1\sim(i,r')$ means that, row (and column) index $t_1$ of the alignment matrix $A$ corresponds to the mapping edge $(i,r')$. Since $G_1$ and $G_2$ are isomorphic with permutation matrix $P=I$, we have $G_1(i,j)=G_2(i',j')$. Therefore, equation \eqref{eq:alignment-erdos} can be written as,
\begin{align}\label{eq:alignment-erdos2}
&A(t_1,t_2)=\\
&\begin{cases}
    (\alpha+1)G_1(i,j)^2-2G_1(i,j)+1+\eps\\
    \text{if } t_1\sim(i,i'),t_2\sim(j,j'),t_1\ \text{and}\ t_2\in \cS_1, t_1\neq t_2.\\
    (\alpha+1)G_1(i,j)G_1(r,s)-G_1(i,j)-G_1(r,s)+1+\eps\\
    \text{if } t_1\sim(i,r'),t_2\sim(j,s'), t_1\ \text{or}\ t_2\in \cS_2, t_1\neq t_2.\\
    1+\eps,\\
    \text{if } t_1=t_2.
\end{cases}\nonumber
\end{align}
Let $\barA$ be the expected alignment matrix, where $\barA(t_1,t_2)=\EE[A(t_1,t_2)]$, the expected value of $A(t_1,t_2)$ over different realizations of $G_1$ and $G_2$.
\begin{lemma}\label{lem:average}
Let $\bv$ be the eigenvector of the expected alignment matrix $\barA$ corresponding to the largest eigenvalue. Suppose
$$
\bv=\begin{lbmatrix}{1}
  \bv_1 \\
  \hline
  \bv_2
\end{lbmatrix},
$$
where $\bv_1$ and $\bv_2$ are vectors of length $n$ and $(k-1)n$, respectively. Then,
\begin{eqnarray}
&&v_{1,1}=v_{1,2}=\ldots=v_{1,n}\triangleq v_1^*,\nonumber\\
&&v_{2,1}=v_{2,2}=\ldots=v_{2,(k-1)n}\triangleq v_2^*,\nonumber
\end{eqnarray}
Moreover, if $n\to \infty$, then,
\begin{equation}
\frac{v_1^*}{v_2^*}> 1+\Delta,
\end{equation}
where $0<\Delta k<\frac{(\alpha-1)p+1+\eps}{(\alpha+1)p^2-2p+1+\eps}-1$.
\end{lemma}
\begin{proof}
Since $G_1(i,j)$ is a Bernoulli random variable which is one with probability $p$, equation \eqref{eq:alignment-erdos2} leads to:
\begin{equation}\label{eq:alignment-erdos22}
\barA(t_1,t_2)=\begin{cases}
    (\alpha-1)p+1+\eps,& \text{if } t_1\ \text{and}\ t_2\in \cS_1, t_1\neq t_2,\\
    (\alpha+1)p^2-2p+1+\eps,& \text{if } t_1\ \text{or}\ t_2\in \cS_2, t_1\neq t_2,\\
    1+\eps, & \text{if } t_1=t_2.
\end{cases}
\end{equation}
Define $a\triangleq (\alpha-1)p+1+\eps$ and $b\triangleq(\alpha+1)p^2-2p+1+\eps$. Since $bv$ is an eigenvector of $\barA$, we have,
\begin{equation}
\barA \bv= \lambda \bv,
\end{equation}
where $\lambda$ is the corresponding eigenvalue of $bv$. Therefore,
\begin{equation}
\barA \bv=\begin{lbmatrix}{1}
  a \sum_{i} v_{1,i}+b\sum_{j} v_{2,j}+(1+\eps-a)v_{1,1} \\
  \vdots\\
  a \sum_{i} v_{1,i}+b\sum_{j} v_{2,j}+(1+\eps-a)v_{1,n}\\
  \hline
  b \sum_{i} v_{1,j}+b\sum_{j} v_{2,j}+(1+\eps-a)v_{2,1}\\
  \vdots\\
  b \sum_{i} v_{1,i}+b\sum_{j} v_{2,j}+(1+\eps-a)v_{2,(k-1)n}
\end{lbmatrix}=\lambda \begin{lbmatrix}{1}
  v_{1,1} \\
  \vdots\\
  v_{1,n}\\
  \hline
  v_{2,1}\\
  \vdots\\
  v_{2,(k-1)n}
\end{lbmatrix}.
\end{equation}
Therefore,
\begin{align}\label{eq:lemma1}
&a\sum_{i}v_{1,i}+b\sum_{j}v_{2,j}=v_{1,r} (\lambda+a-1-\eps),\quad\quad\forall 1 \leq r\leq n, \\
&b\sum_{i}v_{1,i}+b\sum_{j}v_{2,j}=v_{2,s} (\lambda+b-1-\eps),\quad\quad\forall 1 \leq s\leq (k-1)n.\nonumber
\end{align}
We choose $\eps$ so that $\lambda+a-1-\eps\neq 0$ and $\lambda+b-1-\eps\neq 0$. We will show later in this section that any sufficiently small value of $\eps$ satisfies these inequalities. Therefore, equation \eqref{eq:lemma1} leads to,
\begin{eqnarray}\label{eq:lemma2}
&&v_{1,1}=v_{1,2}=\ldots=v_{1,n}=v_1^*,\\
&&v_{2,1}=v_{2,2}=\ldots=v_{2,(k-1)n}=v_2^*.\nonumber
\end{eqnarray}
Using equations (\ref{eq:lemma1}) and (\ref{eq:lemma2}), we have,
\begin{equation}\label{eq:lemma3}
\begin{cases}
    anv_1^*+b(k-1)nv_2^*=v_1^*(\lambda+a-1-\eps)& \\
    bnv_1^*+b(k-1)nv_2^*=v_2^*(\lambda+b-1-\eps).&
\end{cases}
\end{equation}
We choose $\eps$ so that $\lambda+b\big(1-(k-1)n\big)-1-\eps\neq 0$. We will show later in this section that any sufficiently small value of $\eps$ satisfies this inequality. Further, according to the Perron-Frobenius Theorem \ref{thm:perron}, $v_{1,i}>0$ and $v_{2,j}>0$, for all $i$ and $j$. Under these conditions, solving equation (\ref{eq:lemma3}) leads to:
\begin{equation}\label{eq:lemma4}
(\lambda-\lambda_a)(\lambda-\lambda_b)=b^2(k-1)n^2,
\end{equation}
where,
\begin{equation}\label{eq:lemma5}
\begin{cases}
    \lambda_a=(n-1)a+1+\eps,&\\
    \lambda_b=\big((k-1)n-1\big)b+1+\eps.&
\end{cases}
\end{equation}
Equation \eqref{eq:lemma4} has two solutions for $\lambda$. However, since $\lambda$ is the largest eigenvalue of the expected alignment matrix $\barA$, we choose the largest solution. Note that, since $b^2(k-1)n^2>0$, we have $\lambda > \max(\lambda_a,\lambda_b)$. This guarantees conditions that we put on $\eps$ in the early steps of the proof.

By solving equations \eqref{eq:lemma4} and \eqref{eq:lemma5}, we have,
\begin{equation}
\lambda=\frac{\lambda_a+\lambda_b+\sqrt{(\lambda_a-\lambda_b)^2+4(k-1)b^2n^2}}{2}.\nonumber
\end{equation}
First, we show $v_1^*>v_2^*$. As $n\to\infty$, equation \eqref{eq:lemma3} implies,
\begin{equation}\label{eq:lemma6}
\frac{v_1^*}{v_2^*}=\frac{\lambda}{bn}-k+1,
\end{equation}
where $\lambda$ is the largest root of equation \eqref{eq:lemma4}. For sufficiently large $n$,
\begin{equation}\label{eq:lemma7}
\frac{v_1^*}{v_2^*}=\frac{1}{2}[(\frac{a}{b}-k+1)+\sqrt{(\frac{a}{b}-k+1)^2+4k-4}].
\end{equation}
If $p\neq 0,1$, we always have $a>b$. Therefore, there exists $\Delta>0$ such that $\frac{a}{b}>1+\Delta k$. Thus, we have,
\begin{equation}\label{eq:lemma8}
\frac{a}{b}>1+\Delta k>1+\Delta(1+\frac{k-1}{1+\Delta})=1+\Delta+\frac{\Delta}{\Delta+1} (k-1).
\end{equation}
Using inequality \eqref{eq:lemma8} in \eqref{eq:lemma7}, we have,
\begin{align}
\frac{v_1^*}{v_2^*}>\frac{1}{2}\bigg[&\frac{(1+\Delta)^2-k+1}{1+\Delta}\\
&+\frac{\sqrt{((1+\Delta)^2-k+1)^2+4(k-1)(1+\Delta)^2}}{1+\Delta}\bigg]=1+\Delta.\nonumber
\end{align}
This completes the proof of Lemma \ref{lem:average}.
\end{proof}

\subsection{Proof of Theorem \ref{thm:erdos-noisy}} \label{subsec:proof-erdos-noisy}
Without loss of generality and similarly to the proof of Theorem \ref{thm:erdos}, let $P=I$. Let $A$ be the alignment graph of $G_1$ and $G_2$ defined according to equation \eqref{eq:alignment-net}. Similarly to the proof of Theorem \ref{thm:erdos}, re-order row (and column) indices of matrix $A$ so that the first $n$ indices correspond to the true mappings $\{(i,i'):i\in\cV_1, i'\in\cV_2\}$. Define the expected alignment graph $\barA$ as $\barA(t_1,t_2)=\EE[A(t_1,t_2)]$, where $t_1$ and $t_2$ are two possible mappings across graphs. Recall notations $\cS_1=\{1,2,\ldots,n\}$ and $\cS_2=\{n+1,n+2,\ldots,kn\}$.

First, we consider the noise model I \eqref{eq:noise-model1}. Define,
\begin{align}\label{eq:a-b-noisy}
a' \triangleq& p(1-p_e)(\alpha+\eps)+(1-p)(1-p_e) (1+\eps)+ (pp_e+(1-p)p_e)\eps \\
b' \triangleq& \big(p^2(1-p_e)+p p_e(1-p)\big)(\alpha+\eps)\nonumber\\
+& \big((1-p)^2(1-p_e)+p p_e (1-p)\big) (1+\eps)\nonumber\\
+& \big( 2 p (1-p) (1-p_e)+ 2 p^2 p_e \big) \eps.\nonumber
\end{align}
Since $G_1(i,j)$ and $Q(i,j)$ are Bernoulli random variables with parameters $p$ and $p_e$, respectively, the expected alignment graph can be simplified as follows:
\begin{equation}\label{eq:alignment-noisy}
\barA(t_1,t_2)=\begin{cases}
    a',& \text{if } t_1\ \text{and}\ t_2\in \cS_1, t_1\neq t_2,\\
    b',& \text{if } t_1\ \text{or}\ t_2\in \cS_2, t_1\neq t_2,\\
    1+\eps, & \text{if } t_1=t_2.
\end{cases}
\end{equation}
We have,
\begin{align}
a'-b'= (\alpha+1)(2p_e-1)p(p-1)+p_e(1-2p)\eps.
\end{align}
Thus, if $p\neq 0,1$ and $p_e<1/2$, for small enough $\eps$, $a'>b'>0$. Therefore, there exists a positive $\Delta$ such that $\frac{a'}{b'}=1+\Delta$. The rest of the proof is similar to the one of Theorem \ref{thm:erdos}.

The proof for the noise model II of \eqref{eq:noise-model2} is similar. To simplify notation and illustrate the main idea, here we assume $\eps$ is sufficiently small with negligible effects. Define,
\begin{align}\label{eq:a-b-noisy-model2}
a'' \triangleq& p(1-p_e)(\alpha)+(1-p)(1-\pet)= 1-p\big(1+\alpha(p_e-1)+p_e\big) \\
b'' \triangleq& p^2(1-p_e)\alpha+(1-p)^2(1-\pet)+2p(1-p)\pet(1+\alpha)\nonumber\\
=& 1-p(2+p_e)+p^2(1+\alpha+2p_e).\nonumber
\end{align}
The expected alignment graph in this case is:
\begin{equation}\label{eq:alignment-noisy22}
\barA(t_1,t_2)=\begin{cases}
    a'',& \text{if } t_1\ \text{and}\ t_2\in \cS_1, t_1\neq t_2,\\
    b'',& \text{if } t_1\ \text{or}\ t_2\in \cS_2, t_1\neq t_2,\\
    1+\eps, & \text{if } t_1=t_2.
\end{cases}
\end{equation}
Moreover, we have,
\begin{align}
a''-b''= p\big((1-p-p_e)(1+\alpha)+p_e(1-2p)\big).
\end{align}
If $p<1/2$ and $p_e<1/2$, then $a''-b''>0$. The rest of the proof is similar to the previous case.

\subsection{Proof Of Theorem \ref{thm:gaurantee}}\label{subsec:proof-thm-gaurantee}
By writing Taylor's expansion of $Tr(G_1 X G_2 X^T)$ around the point $X_0$, we have,
\begin{align}
Tr (G_1 X G_2 X^T)=& Tr (G_1 X_0 G_2 X_0^T)+ 2Tr (G_1 X_0 G_2 (X-X_0)^T)\\
&+ Tr (G_1 (X-X_0) G_2 (X-X_0)^T).\nonumber
\end{align}
Let $\Delta\triangleq X-X_0$. Thus, we have,
\begin{align}\label{eq:bound1}
f(X)= \tf(X)+ Tr (G_1 \Delta G_2 \Delta^T).
\end{align}
Let $\sigma_i(G)$ be the $i$-th largest singular value of matrix $G$.
\begin{theorem}[Von Neumann's trace inequality]\label{thm:von-neumann}
Suppose $A$ and $B$ are two $n \times n$ complex matrices. We have,
\begin{align}
|Tr (AB)|\leq \sum_{i=1}^n \sigma_i(A) \sigma_i(B).
\end{align}
\end{theorem}
\begin{proof}
See Theorem 1 of \cite{mirsky1975trace}.
\end{proof}

Using Theorem \ref{thm:von-neumann}, we have,
\begin{align}\label{eq:bound2}
|Tr (G_1 \Delta G_2 \Delta^T)|\leq \sum_{i=1}^{n}\sigma_i(G_1 \Delta) \sigma_i(G_1 \Delta^T).
\end{align}
Moreover, we have,
\begin{align}\label{eq:bound3}
\sigma_i(G_1\Delta) &\leq \min\{\sigma_k(G_1)\sigma_{i+1-k}(\Delta): 1\leq k\leq i\}\\
&\leq \sigma_i(G_1)\sigma_{1}(\Delta).\nonumber
\end{align}
Moreover,
\begin{align}\label{eq:bound4}
\sigma_i(G_1\Delta) &\geq \max\{\sigma_k(G_1)\sigma_{n+i-k}(\Delta): i\leq k\leq n\},\\
&\geq \sigma_i(G_1)\sigma_{n}(\Delta).\nonumber
\end{align}
Moreover, since $\|\Delta\|_{2}\leq \epsilon$, we have $|\sigma_{i}(\Delta)|\leq \epsilon$, for $1\leq i\leq n$. Using \eqref{eq:bound3} and \eqref{eq:bound4}, we have,
\begin{align}\label{eq:bound5}
|\sigma_{i}(G_1\Delta)|\leq \epsilon \sigma_{i}(G_1).
\end{align}
Similarly, we have,
\begin{align}\label{eq:bound6}
|\sigma_{i}(G_2\Delta^T)|\leq \epsilon \sigma_{i}(G_2).
\end{align}
Thus, using \eqref{eq:bound2} and \eqref{eq:bound6}, we have,
\begin{align}\label{eq:bound7}
|Tr (G_1 \Delta G_2 \Delta^T)|\leq \epsilon^2 \sum_{i=1}^{n}\sigma_i(G_1) \sigma_i(G_2).
\end{align}
Using \eqref{eq:bound1} and \eqref{eq:bound7}, we have,
\begin{align}
|f(X)-\tf(X)|\leq \epsilon^2 \sum_{i=1}^{n}\sigma_i(G_1) \sigma_i(G_2).
\end{align}
Moreover, since $X^*$ and $X_{lin}^*$ are optimal solutions of optimizations \eqref{eq:qap-overlap-perm} and \eqref{eq:qap-overlap-linear}, respectively, we have,
\begin{align}
f(X^*) & \leq \tf(X^*)+ \epsilon^2 \sum_{i=1}^{n}\sigma_i(G_1) \sigma_i(G_2)\\
& \leq \tf(X_{lin}^*)+ \epsilon^2 \sum_{i=1}^{n}\sigma_i(G_1) \sigma_i(G_2)\nonumber\\
f(X^*) & \geq f(X_{lin}^*) \geq \tf(X_{lin}^*)- \epsilon^2 \sum_{i=1}^{n}\sigma_i(G_1) \sigma_i(G_2).\nonumber
\end{align}
This completes the proof.

\newpage
\noindent
{\bfseries \huge Appendix}

\appendix

\section{Inference of Regulatory Networks}
In this section we leverage genome-wide functional genomics datasets from ENCODE and modENCODE consortia to infer regulatory networks across human, fly, and worm. In Section 6 in the main text, we compare the structure of these inferred networks using network alignment techniques.

The temporal and spatial expression of genes is coordinated by a hierarchy of transcription factors (TFs), whose interactions with each other and with their target genes form directed regulatory networks\cite{int1}. In addition to individual interactions, the structure of a regulatory network captures a broad systems-level view of regulatory and functional processes, since genes cluster into modules that perform similar functions \cite{int2,int3,int4}.  Accurate inference of these regulatory networks is important both in the recovery and functional characterization of gene modules, and for comparative genomics of regulatory networks across multiple species \cite{int6,int7}. This is especially important because animal genomes, as fly, worm, and mouse are routinely used as models for human disease \cite{int8,int9}.

Here, we infer regulatory networks of human, and model organisms {\it D. melanogaster} fly, and {\it C. elegans} worm, three of the most distant and deeply studied metazoan species. To infer regulatory interactions among transcription factors and target genes in each species, we combine genome-wide transcription factor binding profiles, conserved sequence motif instances \cite{int10} and gene expression levels \cite{int11,int12} for multiple cell types that have been collected by the ENCODE and modENCODE consortia. The main challenge is to integrate these diverse evidence sources of gene regulation in order to infer robust and accurate regulatory interactions for each species.

Ideally, inference of regulatory networks would involve performing extensive all-against-all experiments of chromatin immune-precipitation (ChIP) assays for every known transcription factor in every cell type of an organism, in order to identify all potential targets of TFs, followed by functional assays to verify that a TF-gene interaction is functional \cite{int4, int13}. However, the combinatorial number of pairs of TFs and cell types makes this experiment prohibitively expensive, necessitating the use of methods to reduce dimensionality of this problem. Here, we first infer three types of feature-specific regulatory connections based on functional and physical evidences and then integrate them to infer regulatory interactions in each species (Figure \ref{fig:inference-framework}-a). One feature-specific network is based on using sequence motifs to scan the genome for instances of known binding sites of each TF, and then match predicted binding instances to nearby genes (a motif network). A second approach is to map TFs to genes nearby their ChIP peaks using a window-based approach (a ChIP binding network). The third feature specific network uses gene expression profiles under different conditions in order to find groups of genes that are correlated in expression and therefore likely to function together (an expression-based network).

Previous work \cite{int4} has shown that, while ChIP networks are highly informative of true regulatory interactions, the number of experiments that can be carried out is typically very small, yielding a small number of high confidence interactions.  Motif networks tend to be less informative than ChIP networks, but yield more coverage of the regulatory network, while co-expression based networks tend to include many false-positive edges and are the least informative \cite{int13,int14}. However, integration of these three networks \cite{int3,int15,int16, int17} into a combined network yield better performance than the individual networks in terms of recovering known regulatory interactions, by predicting interactions that tend to be supported by multiple lines of evidence. Here, we use an integration approach that combines interaction ranks across networks \cite{int17}. Inferred regulatory interactions show significant overlap with known interactions in human and fly, indicating the accuracy and robustness of the used inference pipeline. In the following, we explain our network inference framework with more details.
\subsection{Inference of feature specific regulatory networks}\label{subsec:inference-ind-nets}
For each species, we form feature-specific regulatory networks using functional (gene expression profiles) and physical (motif sequences and ChIP peaks) evidences as follows:

{\bf Functional association networks.} Expression-based networks represent interactions among TFs and target genes which are supported by correlation in gene expression levels across multiple samples \cite{int1,int21,int22,int23}. There are several methods to infer regulatory networks using gene expression profiles \cite{int17}. The input for these algorithms is a gene by condition matrix of expression values. The output of these methods are expression-based regulatory networks. We use the side information of TF lists to remove outgoing edges from target genes (in fact, TF lists are used as inputs to network inference algorithms to enhance their performance by limiting search space of the methods.).

To reduce bias and obtain a single expression-based network for each species, we combine results of two different expression-based network inference methods (Figure \ref{fig:inference-framework}-a): one method is CLR \cite{int11} (context likelihood of relatedness) which constructs expression networks using mutual information among gene expression profiles along with a correction step to eliminate background correlations. The second method used is GENIE3 \cite{int12} (Gene Network Inference with Ensemble of Trees) which is a tree-based ensemble method that decomposes the network inference problem to several feature selection subproblems. In each subproblem, it identifies potential regulators by performing a regression analysis using random forest. GENEI3 has been recognized as the top-performing expression based inference method in the DREAM5 challenge \cite{int17}.

Table \ref{tab:exp} summarizes the number of genes and TFs in expression-based regulatory networks. These numbers refer to genes and TFs that are mapped to Entrez Gene IDs \cite{int24}, the standard IDs that we use throughout our analysis. As it is illustrated in this table, expression based networks cover most of potential regulatory edges from TFs to targets. Despite a high coverage, however, the quality of inferred expression networks are lower than the one for physical networks \cite{int13}. This can be partially owing to indirect effects and transitive interactions in expression-based regulatory networks \cite{int17}.

{\bf Physical association networks.} We form two physical regulatory networks for each of the considered species using two types of physical evidences as our inference features: In the first approach, we use conserved occurrences of known sequence motifs \cite{int10}, while in the second approach, we use experimentally defined TF binding occupancy profiles from ChIP assays of ENCODE and modENCODE \cite{int4,int13}. Table \ref{tab:chip} shows the number of TFs associated to motifs as well as the number of TFs with genome-wide ChIP profiles in human, fly and worm. TSS coordinates are based on the genome annotations from ENCODE and modENCODE for human and worm, respectively, and the FlyBase genome annotations (FB5.48) for fly.

\begin{figure*}[t]
  \centering
      \includegraphics[width=12cm]{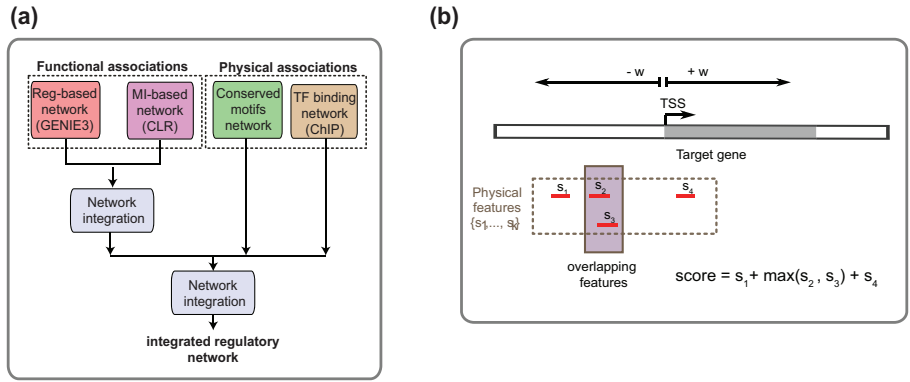}
  \caption{(a) The proposed framework to infer integrative regulatory networks. (b) The proposed framework to infer physical feature-specific regulatory networks. }
  \label{fig:inference-framework}
\end{figure*}

Each physical feature is assigned to a score: motif sequence features are assigned to conservation scores according to a phylogenetic framework \cite{int10}, while sequence read density of TFs determines scores of ChIP peaks. Further, two physical features are called overlapping if their corresponding sequences have a minimum overlap of $25\%$ in relation to their lengths (Jaccard Index $>0.25$).

Our inference algorithm is based on occurrence of these features (motif sequences or ChIP peaks) within a fixed window size around the transcription start site (TSS) of target genes (Figure \ref{fig:inference-framework}-b). We use a fixed window of 5kb around the transcription start site (TSS) in human and 1kb in fly and worm. Then, we apply a max-sum algorithm to assign weights to TF-target interactions in each case: we take the maximum score of overlapping features and sum the scores of non-overlapping ones. In ChIP networks, because read densities are not comparable across different TFs, we normalize TF-target weights for each TF by computing z-scores.

\subsection{Inference of integrated regulatory networks}\label{subsec:inferece-integrative}
Feature specific networks have certain biases and shortcomings. While Physical networks (motif and ChIP networks) show high quality considering overlap of their interactions with known interactions \cite{int13}, their coverage of the entire network is pretty low mostly owing to the cost of the experiments. On the other hand, while expression based networks have a larger coverage of regulatory networks compared to physical ones, they include many false-positive edges partially owing to indirect information flows \cite{int14}. To overcome these limitations, we therefore integrate these feature-specific regulatory interactions into a single integrated network \cite{int3,int15,int16,int17} (Figure \ref{fig:inference-framework}-a).

Suppose there are $K$ input feature-specific regulatory networks, each with $n$ genes and $m$ TFs (only TF nodes can have out-going edges in the network). Let $\wijl$ and $\wij$ represent interaction weights between TF $i$ and target gene $j$ in the input network $l$ and in the integrative network, respectively. We use a rank-based (borda) integration technique to infer integrated networks in considered species. In this approach, integrative weights are computed as follows:
\begin{align}
\wij=1/K \sum_{l=1}^{K} \rijl,
\end{align}
where $\rijl$ represents the rank of interactions $i\to j$ in the input network $l$. An edge with the maximum weight is mapped to the rank $nm$. We also assume non-existent interactions are mapped to rank 0 (if $\wijl=0$, then $\rijl=0$) \cite{int17}. Moreover, ties are broken randomly among edges with same weights.

We find that top-ranking integrative interactions in human and fly networks are primarily supported by ChIP and motif evidences, while worm interactions are primarily supported by co-expression edges, consistent with the lower coverage of worm ChIP and motif interactions (Figure \ref{fig:inference-inp-cont}).

To validate inferred integrated networks, we use known interactions in TRANSFAC \cite{int18}, REDfly \cite{int19} and EdgeDB \cite{int20} as human, fly and worm benchmarks, respectively. We assess the quality of various networks by using (a) the area under the receiver operating characteristic curve (AUROC); and (b) the area under the precision recall curve (AUPR), for each benchmark network (Figures \ref{fig:auroc-aupr}). Let $TP(k)$ and $FP(k)$ represent the number of true positives and false positives in top $k$ predictions, respectively. Suppose the total number of positives and negatives in the gold standard are represented by $P$ and $N$, respectively. Then, an ROC curve plots true positive rate vs. false positive rate ($TP(k)/P$ vs. $FP(k)/N$), while a PR curve plots precision ($TP(k)/k$) vs. recall ($TP(k)/P$). A high AUPR value indicates that, top predictions significantly overlap with known interactions, while a high AUROC value indicates the advantage of inferred predictions in discriminating true and false positives compared to random predictions (AUROC of a random predictor is 0.5).

\begin{figure*}[t]
  \centering
      \includegraphics[width=12cm]{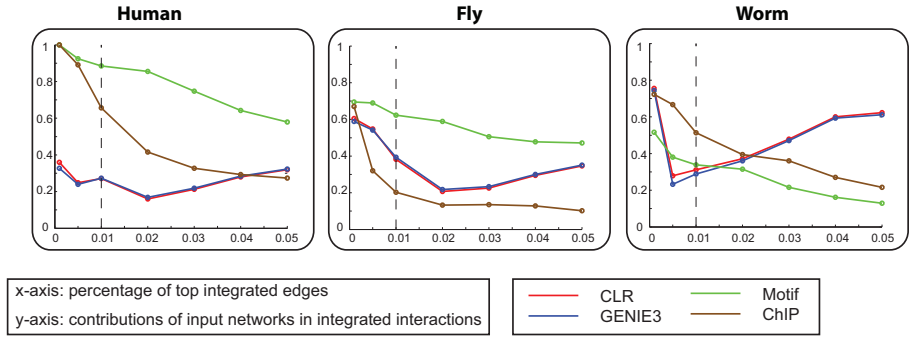}
  \caption{Contributions of input feature-specific networks in integrated interactions.}
  \label{fig:inference-inp-cont}
\end{figure*}

\begin{table}[t!]
\centering
\resizebox{6cm}{!}{
  \begin{tabular}{ | c | c | c | c | }
    \hline
     & \mbox{Human} & \mbox{Fly} & \mbox{Worm} \\ \hline
    \mbox{Genes} & 19,088 & 12,897 & 19,277 \\ \hline
    \mbox{TFs} & 2,757 & 675 & 905 \\ \hline

  \end{tabular}}
\caption{Number of genes and TFs covered by gene expression data.}
\label{tab:exp}
\end{table}

\begin{table}[t!]
\centering
\resizebox{6cm}{!}{
  \begin{tabular}{ | c | c | c | c | }
    \hline
     & \mbox{Human} & \mbox{Fly} & \mbox{Worm} \\ \hline
    \mbox{Motif network} & 485 & 221 & 30 \\ \hline
    \mbox{ChIP network} & 165 & 51 & 88 \\ \hline

  \end{tabular}}
\caption{Number of TFs covered by evolutionary conserved motifs and TF binding datasets.}
\label{tab:chip}
\end{table}

Figure \ref{fig:auroc-aupr} illustrates AUROC and AUPR scores for feature-specific and integrative networks, in different cut-offs, and in all three considered species. Considering the top $5\%$ of interactions in each weighted network as predicted edges, according to AUROC metric, integrative networks outperform feature-specific networks in all three species. In fact, AUROC values of integrative networks are 0.58 in human, 0.62 in fly, and 0.52 in worm, respectively. AUPR values of integrative networks are 0.019 in human, 0.047 in fly, and 0.037 in worm, respectively. Notably, all methods have low scores over the EdgeDB (worm) benchmark, which can be partially owing to sparse physical networks and/or systematic bias of EdgeDB interactions.

As the cut-off (network density) increases, AUROC values of integrative networks tend to increase while their AUPR scores are decreasing in general. This is because of the fact that, the rate of true positives is lower among medium ranked interactions compared to top ones. Considering both AUROC and AUPR curves for all species, we binarize networks using their top $5\%$ interactions which leads to balanced values of AUROC and AUPR in all inferred networks. This results in $2.6M$ interactions in human, $469k$ in fly and $876k$ in worm. In integrative networks, the median number of targets for each TF is 253 in human, 290 in fly and 640 in worm, with a median of 132 regulators per gene in human, 29 in fly, and 43 in worm.

Unlike EigenAlign, other considered network alignment methods do not take into account the directionality of edges in their network alignment setup. Thus, to have fair performance assessments of considered network alignment methods, we create un-directed co-regulatory networks using inferred regulatory networks by connecting genes when their parent TFs have an overlap larger than $25\%$. This results in undirected binary networks in human, fly, and worm, with $19,221$, $13,642$, and $19,296$ nodes, and $13.9\%$, $3.5\%$, and $4.2\%$ edge densities, respectively.

\begin{figure*}[t]
  \centering
      \includegraphics[width=12cm]{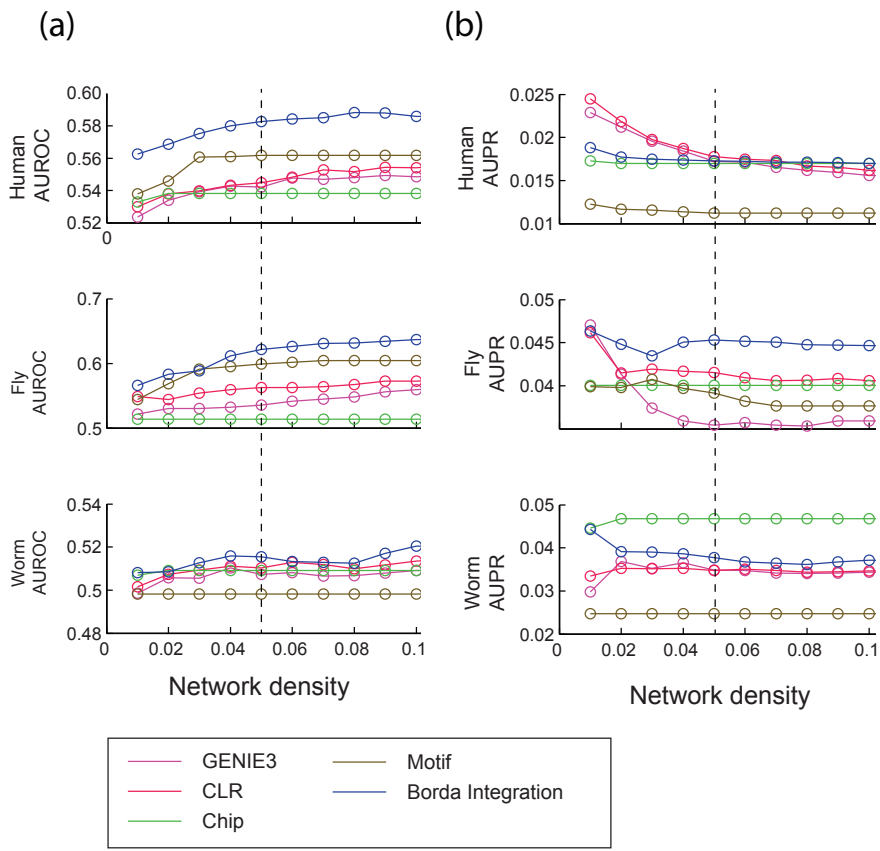}
  \caption{AUROC and AUPR scores of feature-specific and integrated regulatory networks in human, fly and worm species.}
  \label{fig:auroc-aupr}
\end{figure*}



\begin{thebibliography}{10}
\providecommand{\url}[1]{#1}
\csname url@samestyle\endcsname
\providecommand{\newblock}{\relax}
\providecommand{\bibinfo}[2]{#2}
\providecommand{\BIBentrySTDinterwordspacing}{\spaceskip=0pt\relax}
\providecommand{\BIBentryALTinterwordstretchfactor}{4}
\providecommand{\BIBentryALTinterwordspacing}{\spaceskip=\fontdimen2\font plus
\BIBentryALTinterwordstretchfactor\fontdimen3\font minus
  \fontdimen4\font\relax}
\providecommand{\BIBforeignlanguage}[2]{{%
\expandafter\ifx\csname l@#1\endcsname\relax
\typeout{** WARNING: IEEEtran.bst: No hyphenation pattern has been}%
\typeout{** loaded for the language `#1'. Using the pattern for}%
\typeout{** the default language instead.}%
\else
\language=\csname l@#1\endcsname
\fi
#2}}
\providecommand{\BIBdecl}{\relax}
\BIBdecl

\bibitem{Trey2006modeling}
R.~Sharan and T.~Ideker, ``Modeling cellular machinery through biological
  network comparison,'' \emph{Nature biotechnology}, vol.~24, no.~4, pp.
  427--433, 2006.

\bibitem{book1976graphbondy}
J.~A. Bondy and U.~S.~R. Murty, \emph{Graph theory with applications}.\hskip
  1em plus 0.5em minus 0.4em\relax Macmillan London, 1976, vol.~6.

\bibitem{isorank}
R.~Singh, J.~Xu, and B.~Berger, ``Global alignment of multiple protein
  interaction networks with application to functional orthology detection,''
  \emph{Proceedings of the National Academy of Sciences}, vol. 105, no.~35, pp.
  12\,763--12\,768, 2008.

\bibitem{isorankn}
C.-S. Liao, K.~Lu, M.~Baym, R.~Singh, and B.~Berger, ``Isorankn: spectral
  methods for global alignment of multiple protein networks,''
  \emph{Bioinformatics}, vol.~25, no.~12, pp. i253--i258, 2009.

\bibitem{graemlin-flannick2006}
J.~Flannick, A.~Novak, B.~S. Srinivasan, H.~H. McAdams, and S.~Batzoglou,
  ``Graemlin: general and robust alignment of multiple large interaction
  networks,'' \emph{Genome research}, vol.~16, no.~9, pp. 1169--1181, 2006.

\bibitem{graphmatching2009zaslavskiy}
M.~Zaslavskiy, F.~Bach, and J.-P. Vert, ``Global alignment of protein--protein
  interaction networks by graph matching methods,'' \emph{Bioinformatics},
  vol.~25, no.~12, pp. i259--1267, 2009.

\bibitem{pathblast-kelley2004}
B.~P. Kelley, B.~Yuan, F.~Lewitter, R.~Sharan, B.~R. Stockwell, and T.~Ideker,
  ``Pathblast: a tool for alignment of protein interaction networks,''
  \emph{Nucleic acids research}, vol.~32, no. suppl 2, pp. W83--W88, 2004.

\bibitem{networkblast-kalaev2008}
M.~Kalaev, M.~Smoot, T.~Ideker, and R.~Sharan, ``Networkblast: comparative
  analysis of protein networks,'' \emph{Bioinformatics}, vol.~24, no.~4, pp.
  594--596, 2008.

\bibitem{visionconte2004thirty}
D.~Conte, P.~Foggia, C.~Sansone, and M.~Vento, ``Thirty years of graph matching
  in pattern recognition,'' \emph{International journal of pattern recognition
  and artificial intelligence}, vol.~18, no.~03, pp. 265--298, 2004.

\bibitem{visionschellewald2005probabilistic}
C.~Schellewald and C.~Schn{\"o}rr, ``Probabilistic subgraph matching based on
  convex relaxation,'' in \emph{Energy minimization methods in computer vision
  and pattern recognition}.\hskip 1em plus 0.5em minus 0.4em\relax Springer,
  2005, pp. 171--186.

\bibitem{ontology-lacoste2006word}
S.~Lacoste-Julien, B.~Taskar, D.~Klein, and M.~I. Jordan, ``Word alignment via
  quadratic assignment,'' in \emph{Proceedings of the main conference on Human
  Language Technology Conference of the North American Chapter of the
  Association of Computational Linguistics}.\hskip 1em plus 0.5em minus
  0.4em\relax Association for Computational Linguistics, 2006, pp. 112--119.

\bibitem{ontology-melnik2002similarity}
S.~Melnik, H.~Garcia-Molina, and E.~Rahm, ``Similarity flooding: A versatile
  graph matching algorithm and its application to schema matching,'' in
  \emph{Data Engineering, 2002. Proceedings. 18th International Conference
  on}.\hskip 1em plus 0.5em minus 0.4em\relax IEEE, 2002, pp. 117--128.

\bibitem{twitter-deanon}
A.~Narayanan and V.~Shmatikov, ``De-anonymizing social networks,'' in
  \emph{Security and Privacy, 2009 30th IEEE Symposium on}.\hskip 1em plus
  0.5em minus 0.4em\relax IEEE, 2009, pp. 173--187.

\bibitem{quadratic-book-burkard2013}
R.~E. Burkard, \emph{Quadratic assignment problems}.\hskip 1em plus 0.5em minus
  0.4em\relax Springer, 2013.

\bibitem{QAP-newlp-NPhard}
K.~Makarychev, R.~Manokaran, and M.~Sviridenko, ``Maximum quadratic assignment
  problem: Reduction from maximum label cover and lp-based approximation
  algorithm,'' in \emph{Automata, Languages and Programming}.\hskip 1em plus
  0.5em minus 0.4em\relax Springer, 2010, pp. 594--604.

\bibitem{branch-and-bound}
M.~Bazaraa and O.~Kirca, ``A branch-and-bound-based heuristic for solving the
  quadratic assignment problem,'' \emph{Naval research logistics quarterly},
  vol.~30, no.~2, pp. 287--304, 1983.

\bibitem{cuttingplane}
M.~S. Bazaraa and H.~D. Sherali, ``On the use of exact and heuristic cutting
  plane methods for the quadratic assignment problem,'' \emph{Journal of the
  Operational Research Society}, pp. 991--1003, 1982.

\bibitem{lawler1963quadratic}
E.~L. Lawler, ``The quadratic assignment problem,'' \emph{Management science},
  vol.~9, no.~4, pp. 586--599, 1963.

\bibitem{kaufman1978algorithm}
L.~Kaufman and F.~Broeckx, ``An algorithm for the quadratic assignment problem
  using bender's decomposition,'' \emph{European Journal of Operational
  Research}, vol.~2, no.~3, pp. 207--211, 1978.

\bibitem{frieze1983quadratic}
A.~Frieze and J.~Yadegar, ``On the quadratic assignment problem,''
  \emph{Discrete applied mathematics}, vol.~5, no.~1, pp. 89--98, 1983.

\bibitem{adams1994improved}
W.~P. Adams and T.~A. Johnson, ``Improved linear programming-based lower bounds
  for the quadratic assignment problem,'' \emph{DIMACS series in discrete
  mathematics and theoretical computer science}, vol.~16, pp. 43--75, 1994.

\bibitem{orthogonal-finke1987quadratic}
G.~Finke, R.~E. Burkard, and F.~Rendl, ``Quadratic assignment problems,''
  \emph{North-Holland Mathematics Studies}, vol. 132, pp. 61--82, 1987.

\bibitem{projection-hadley1992new}
S.~Hadley, F.~Rendl, and H.~Wolkowicz, ``A new lower bound via projection for
  the quadratic assignment problem,'' \emph{Mathematics of Operations
  Research}, vol.~17, no.~3, pp. 727--739, 1992.

\bibitem{convex-anstreicher2000lagrangian}
K.~Anstreicher and H.~Wolkowicz, ``On lagrangian relaxation of quadratic matrix
  constraints,'' \emph{SIAM Journal on Matrix Analysis and Applications},
  vol.~22, no.~1, pp. 41--55, 2000.

\bibitem{convex-anstreicher2001solving}
K.~M. Anstreicher and N.~W. Brixius, ``Solving quadratic assignment problems
  using convex quadratic programming relaxations,'' \emph{Optimization Methods
  and Software}, vol.~16, no. 1-4, pp. 49--68, 2001.

\bibitem{semidefinite-zhao1998}
Q.~Zhao, S.~E. Karisch, F.~Rendl, and H.~Wolkowicz, ``Semidefinite programming
  relaxations for the quadratic assignment problem,'' \emph{Journal of
  Combinatorial Optimization}, vol.~2, no.~1, pp. 71--109, 1998.

\bibitem{splitting-sdp}
J.~Peng, H.~Mittelmann, and X.~Li, ``A new relaxation framework for quadratic
  assignment problems based on matrix splitting,'' \emph{Mathematical
  Programming Computation}, vol.~2, no.~1, pp. 59--77, 2010.

\bibitem{vogelstein2015fast}
J.~T. Vogelstein, J.~M. Conroy, V.~Lyzinski, L.~J. Podrazik, S.~G. Kratzer,
  E.~T. Harley, D.~E. Fishkind, R.~J. Vogelstein, and C.~E. Priebe, ``Fast
  approximate quadratic programming for graph matching,'' \emph{PLOS one},
  vol.~10, no.~4, 2015.

\bibitem{klau2009new}
G.~W. Klau, ``A new graph-based method for pairwise global network alignment,''
  \emph{BMC bioinformatics}, vol.~10, no. Suppl 1, p. S59, 2009.

\bibitem{natalie}
M.~El-Kebir, J.~Heringa, and G.~W. Klau, ``Lagrangian relaxation applied to
  sparse global network alignment,'' in \emph{IAPR International Conference on
  Pattern Recognition in Bioinformatics}, 2011, pp. 225--236.

\bibitem{natalie2}
------, ``Natalie 2.0: Sparse global network alignment as a special case of
  quadratic assignment,'' \emph{Algorithms}, vol.~8, no.~4, pp. 1035--1051,
  2015.

\bibitem{Bayesiankolavr2012graphalignment}
M.~Kol{\'a}{\v{r}}, J.~Meier, V.~Mustonen, M.~L{\"a}ssig, and J.~Berg,
  ``Graphalignment: Bayesian pairwise alignment of biological networks,''
  \emph{BMC systems biology}, vol.~6, no.~1, p. 144, 2012.

\bibitem{bayati2013message}
M.~Bayati, D.~F. Gleich, A.~Saberi, and Y.~Wang, ``Message-passing algorithms
  for sparse network alignment,'' \emph{ACM Transactions on Knowledge Discovery
  from Data (TKDD)}, vol.~7, no.~1, p.~3, 2013.

\bibitem{vision-spectral-main}
M.~Leordeanu and M.~Hebert, ``A spectral technique for correspondence problems
  using pairwise constraints,'' in \emph{Computer Vision, 2005. ICCV 2005.
  Tenth IEEE International Conference on}, vol.~2.\hskip 1em plus 0.5em minus
  0.4em\relax IEEE, 2005, pp. 1482--1489.

\bibitem{vision-spectral-clustering}
M.~Carcassoni and E.~R. Hancock, ``Alignment using spectral clusters.'' in
  \emph{BMVC}, 2002, pp. 1--10.

\bibitem{kazemi2015growing}
E.~Kazemi, H.~S~Hamed, and M.~Grossglauser, ``Growing a graph matching from a
  handful of seeds,'' in \emph{Proceedings of the Vldb Endowment International
  Conference on Very Large Data Bases}, vol.~8, no. EPFL-ARTICLE-207759, 2015.

\bibitem{clark2015multiobjective}
C.~Clark and J.~Kalita, ``A multiobjective memetic algorithm for ppi network
  alignment,'' \emph{Bioinformatics}, p. btv063, 2015.

\bibitem{malod2015graal}
N.~Malod-Dognin and N.~Pr{\v{z}}ulj, ``L-graal: Lagrangian graphlet-based
  network aligner,'' \emph{Bioinformatics}, p. btv130, 2015.

\bibitem{QAPsurvey-Elsevier}
E.~M. Loiola, N.~M.~M. de~Abreu, P.~O. Boaventura-Netto, P.~Hahn, and
  T.~Querido, ``A survey for the quadratic assignment problem,'' \emph{European
  Journal of Operational Research}, vol. 176, no.~2, pp. 657--690, 2007.

\bibitem{emmert2016fifty}
F.~Emmert-Streib, M.~Dehmer, and Y.~Shi, ``Fifty years of graph matching,
  network alignment and network comparison,'' \emph{Information Sciences}, vol.
  346, pp. 180--197, 2016.

\bibitem{newman2006modularity}
M.~E. Newman, ``Modularity and community structure in networks,''
  \emph{Proceedings of the National Academy of Sciences}, vol. 103, no.~23, pp.
  8577--8582, 2006.

\bibitem{sussman}
D.~L. Sussman, M.~Tang, D.~E. Fishkind, and C.~E. Priebe, ``A consistent
  adjacency spectral embedding for stochastic blockmodel graphs,''
  \emph{Journal of the American Statistical Association}, vol. 107, no. 499,
  pp. 1119--1128, 2012.

\bibitem{qin}
T.~Qin and K.~Rohe, ``Regularized spectral clustering under the
  degree-corrected stochastic blockmodel,'' in \emph{Advances in Neural
  Information Processing Systems}, 2013, pp. 3120--3128.

\bibitem{athreya}
A.~Athreya, V.~Lyzinski, D.~J. Marchette, C.~E. Priebe, D.~L. Sussman, and
  M.~Tang, ``A central limit theorem for scaled eigenvectors of random dot
  product graphs,'' \emph{arXiv preprint arXiv:1305.7388}, 2013.

\bibitem{saade}
A.~Saade, F.~Krzakala, and L.~Zdeborov{\'a}, ``Spectral clustering of graphs
  with the bethe hessian,'' in \emph{Advances in Neural Information Processing
  Systems}, 2014, pp. 406--414.

\bibitem{schellewald2007evaluation}
C.~Schellewald, S.~Roth, and C.~Schn{\"o}rr, ``Evaluation of a convex
  relaxation to a quadratic assignment matching approach for relational object
  views,'' \emph{Image and Vision Computing}, vol.~25, no.~8, pp. 1301--1314,
  2007.

\bibitem{patro2012global}
R.~Patro and C.~Kingsford, ``Global network alignment using multiscale spectral
  signatures,'' \emph{Bioinformatics}, vol.~28, no.~23, pp. 3105--3114, 2012.

\bibitem{erdHos1961strength}
P.~Erd{\H{o}}s and A.~R{\'e}nyi, ``On the strength of connectedness of a random
  graph,'' \emph{Acta Mathematica Hungarica}, vol.~12, no.~1, pp. 261--267,
  1961.

\bibitem{isom-czajka2008improved}
T.~Czajka and G.~Pandurangan, ``Improved random graph isomorphism,''
  \emph{Journal of Discrete Algorithms}, vol.~6, no.~1, pp. 85--92, 2008.

\bibitem{babai2016graph}
L.~Babai, ``Graph isomorphism in quasipolynomial time [extended abstract],'' in
  \emph{Proceedings of the 48th Annual ACM SIGACT Symposium on Theory of
  Computing}.\hskip 1em plus 0.5em minus 0.4em\relax ACM, 2016, pp. 684--697.

\bibitem{clark2014comparison}
C.~Clark and J.~Kalita, ``A comparison of algorithms for the pairwise alignment
  of biological networks,'' \emph{Bioinformatics}, vol.~30, no.~16, pp.
  2351--2359, 2014.

\bibitem{moduleali2009functionally}
W.~Ali and C.~M. Deane, ``Functionally guided alignment of protein interaction
  networks for module detection,'' \emph{Bioinformatics}, vol.~25, no.~23, pp.
  3166--3173, 2009.

\bibitem{mohammadi2015triangular}
S.~Mohammadi, D.~Gleich, T.~Kolda, and A.~Grama, ``Triangular alignment (tame):
  A tensor-based approach for higher-order network alignment,'' \emph{arXiv
  preprint arXiv:1510.06482}, 2015.

\bibitem{complexity-isomorphisim-schweitzer}
P.~Schweitzer, ``Problems of unknown complexity: graph isomorphism and ramsey
  theoretic numbers,'' Ph.D. dissertation, Saarbrücken, Univ., Diss., 2009,
  2009.

\bibitem{babai1979canonical}
L.~Babai and L.~Kucera, ``Canonical labelling of graphs in linear average
  time,'' in \emph{Foundations of Computer Science, 1979., 20th Annual
  Symposium on}.\hskip 1em plus 0.5em minus 0.4em\relax IEEE, 1979, pp. 39--46.

\bibitem{trefethen1997numerical}
L.~N. Trefethen and D.~Bau~III, \emph{Numerical linear algebra}.\hskip 1em plus
  0.5em minus 0.4em\relax Siam, 1997, vol.~50.

\bibitem{pillai2005perron}
S.~U. Pillai, T.~Suel, and S.~Cha, ``The perron-frobenius theorem: some of its
  applications,'' \emph{IEEE Signal Processing Magazine}, vol.~22, no.~2, pp.
  62--75, 2005.

\bibitem{book-graphtheory-west2001introduction}
D.~B. West \emph{et~al.}, \emph{Introduction to graph theory}.\hskip 1em plus
  0.5em minus 0.4em\relax Prentice hall Upper Saddle River, 2001, vol.~2.

\bibitem{complexity-eigen}
J.~Kuczynski and H.~Wozniakowski, ``Estimating the largest eigenvalue by the
  power and lanczos algorithms with a random start,'' \emph{SIAM journal on
  matrix analysis and applications}, vol.~13, no.~4, pp. 1094--1122, 1992.

\bibitem{preis1999linear}
R.~Preis, ``Linear time 1/2-approximation algorithm for maximum weighted
  matching in general graphs,'' in \emph{Annual Symposium on Theoretical
  Aspects of Computer Science}, 1999, pp. 259--269.

\bibitem{fraikin2008gradient}
C.~Fraikin, Y.~Nesterov, and P.~Van~Dooren, ``A gradient-type algorithm
  optimizing the coupling between matrices,'' \emph{Linear Algebra and its
  Applications}, vol. 429, no.~5, pp. 1229--1242, 2008.

\bibitem{fraikin2008optimizing}
------, ``Optimizing the coupling between two isometric projections of
  matrices,'' \emph{SIAM Journal on Matrix Analysis and Applications}, vol.~30,
  no.~1, pp. 324--345, 2008.

\bibitem{knossow2009inexact}
D.~Knossow, A.~Sharma, D.~Mateus, and R.~Horaud, ``Inexact matching of large
  and sparse graphs using laplacian eigenvectors,'' in \emph{International
  workshop on graph-based representations in pattern recognition}, 2009, pp.
  144--153.

\bibitem{orth-qap-alg}
S.~Umeyama, ``An eigendecomposition approach to weighted graph matching
  problems,'' \emph{Pattern Analysis and Machine Intelligence, IEEE
  Transactions on}, vol.~10, no.~5, pp. 695--703, 1988.

\bibitem{cross-validation}
R.~Kohavi \emph{et~al.}, ``A study of cross-validation and bootstrap for
  accuracy estimation and model selection,'' in \emph{IJCAI}, vol.~14, no.~2,
  1995, pp. 1137--1145.

\bibitem{power-law2001random}
W.~Aiello, F.~Chung, and L.~Lu, ``A random graph model for power law graphs,''
  \emph{Experimental Mathematics}, vol.~10, no.~1, pp. 53--66, 2001.

\bibitem{jessica-paper}
A.~P. Boyle, C.~L. Araya, C.~Brdlik, P.~Cayting, C.~Cheng, Y.~Cheng,
  K.~Gardner, L.~W. Hillier, J.~Janette, L.~Jiang \emph{et~al.}, ``Comparative
  analysis of regulatory information and circuits across distant species,''
  \emph{Nature}, vol. 512, no. 7515, pp. 453--456, 2014.

\bibitem{mirsky1975trace}
L.~Mirsky, ``A trace inequality of john von neumann,'' \emph{Monatshefte
  f{\"u}r mathematik}, vol.~79, no.~4, pp. 303--306, 1975.

\bibitem{int1}
R.~De~Smet and K.~Marchal, ``Advantages and limitations of current network
  inference methods,'' \emph{Nature Reviews Microbiology}, vol.~8, no.~10, pp.
  717--729, 2010.

\bibitem{int2}
E.~Segal, M.~Shapira, A.~Regev, D.~Pe'er, D.~Botstein, D.~Koller, and
  N.~Friedman, ``Module networks: identifying regulatory modules and their
  condition-specific regulators from gene expression data,'' \emph{Nature
  genetics}, vol.~34, no.~2, pp. 166--176, 2003.

\bibitem{int3}
Z.~Bar-Joseph, G.~K. Gerber, T.~I. Lee, N.~J. Rinaldi, J.~Y. Yoo, F.~Robert,
  D.~B. Gordon, E.~Fraenkel, T.~S. Jaakkola, R.~A. Young \emph{et~al.},
  ``Computational discovery of gene modules and regulatory networks,''
  \emph{Nature biotechnology}, vol.~21, no.~11, pp. 1337--1342, 2003.

\bibitem{int4}
D.~Marbach, S.~Roy, F.~Ay, P.~E. Meyer, R.~Candeias, T.~Kahveci, C.~A. Bristow,
  and M.~Kellis, ``Predictive regulatory models in drosophila melanogaster by
  integrative inference of transcriptional networks,'' \emph{Genome research},
  vol.~22, no.~7, pp. 1334--1349, 2012.

\bibitem{int6}
R.~Sharan and T.~Ideker, ``Modeling cellular machinery through biological
  network comparison,'' \emph{Nature biotechnology}, vol.~24, no.~4, pp.
  427--433, 2006.

\bibitem{int7}
S.~A. McCarroll, C.~T. Murphy, S.~Zou, S.~D. Pletcher, C.-S. Chin, Y.~N. Jan,
  C.~Kenyon, C.~I. Bargmann, and H.~Li, ``Comparing genomic expression patterns
  across species identifies shared transcriptional profile in aging,''
  \emph{Nature genetics}, vol.~36, no.~2, pp. 197--204, 2004.

\bibitem{int8}
J.~O. Woods, U.~M. Singh-Blom, J.~M. Laurent, K.~L. McGary, and E.~M. Marcotte,
  ``Prediction of gene--phenotype associations in humans, mice, and plants
  using phenologs,'' \emph{BMC bioinformatics}, vol.~14, no.~1, p. 203, 2013.

\bibitem{int9}
V.~R. Chintapalli, J.~Wang, and J.~A. Dow, ``Using flyatlas to identify better
  drosophila melanogaster models of human disease,'' \emph{Nature genetics},
  vol.~39, no.~6, pp. 715--720, 2007.

\bibitem{int10}
P.~Kheradpour, A.~Stark, S.~Roy, and M.~Kellis, ``Reliable prediction of
  regulator targets using 12 drosophila genomes,'' \emph{Genome research},
  vol.~17, no.~12, pp. 1919--1931, 2007.

\bibitem{int11}
J.~J. Faith, B.~Hayete, J.~T. Thaden, I.~Mogno, J.~Wierzbowski, G.~Cottarel,
  S.~Kasif, J.~J. Collins, and T.~S. Gardner, ``Large-scale mapping and
  validation of escherichia coli transcriptional regulation from a compendium
  of expression profiles,'' \emph{PLoS biology}, vol.~5, no.~1, p.~e8, 2007.

\bibitem{int12}
A.~Irrthum, L.~Wehenkel, P.~Geurts \emph{et~al.}, ``Inferring regulatory
  networks from expression data using tree-based methods,'' \emph{PloS one},
  vol.~5, no.~9, p. e12776, 2010.

\bibitem{int13}
S.~Roy, J.~Ernst, P.~V. Kharchenko, P.~Kheradpour, N.~Negre, M.~L. Eaton, J.~M.
  Landolin, C.~A. Bristow, L.~Ma, M.~F. Lin \emph{et~al.}, ``Identification of
  functional elements and regulatory circuits by drosophila modencode,''
  \emph{Science}, vol. 330, no. 6012, pp. 1787--1797, 2010.

\bibitem{int14}
S.~Feizi, D.~Marbach, M.~M{\'e}dard, and M.~Kellis, ``Network deconvolution as
  a general method to distinguish direct dependencies in networks,''
  \emph{Nature biotechnology}, 2013.

\bibitem{int15}
D.~J. Reiss, N.~S. Baliga, and R.~Bonneau, ``Integrated biclustering of
  heterogeneous genome-wide datasets for the inference of global regulatory
  networks,'' \emph{BMC bioinformatics}, vol.~7, no.~1, p. 280, 2006.

\bibitem{int16}
A.~Greenfield, A.~Madar, H.~Ostrer, and R.~Bonneau, ``Dream4: Combining genetic
  and dynamic information to identify biological networks and dynamical
  models,'' \emph{PloS one}, vol.~5, no.~10, p. e13397, 2010.

\bibitem{int17}
D.~Marbach, J.~C. Costello, R.~K{\"u}ffner, N.~M. Vega, R.~J. Prill, D.~M.
  Camacho, K.~R. Allison, M.~Kellis, J.~J. Collins, G.~Stolovitzky
  \emph{et~al.}, ``Wisdom of crowds for robust gene network inference,''
  \emph{Nature methods}, vol.~9, no.~8, pp. 796--804, 2012.

\bibitem{int21}
D.~Marbach, R.~J. Prill, T.~Schaffter, C.~Mattiussi, D.~Floreano, and
  G.~Stolovitzky, ``Revealing strengths and weaknesses of methods for gene
  network inference,'' \emph{Proceedings of the National Academy of Sciences},
  vol. 107, no.~14, pp. 6286--6291, 2010.

\bibitem{int22}
R.~Bonneau, D.~J. Reiss, P.~Shannon, M.~Facciotti, L.~Hood, N.~S. Baliga, and
  V.~Thorsson, ``The inferelator: an algorithm for learning parsimonious
  regulatory networks from systems-biology data sets de novo,'' \emph{Genome
  biology}, vol.~7, no.~5, p. R36, 2006.

\bibitem{int23}
N.~Friedman, M.~Linial, I.~Nachman, and D.~Pe'er, ``Using bayesian networks to
  analyze expression data,'' \emph{Journal of computational biology}, vol.~7,
  no. 3-4, pp. 601--620, 2000.

\bibitem{int24}
D.~Maglott, J.~Ostell, K.~D. Pruitt, and T.~Tatusova, ``Entrez gene:
  gene-centered information at ncbi,'' \emph{Nucleic acids research}, vol.~33,
  no. suppl 1, pp. D54--D58, 2005.

\bibitem{int18}
E.~Wingender, X.~Chen, R.~Hehl, H.~Karas, I.~Liebich, V.~Matys, T.~Meinhardt,
  M.~Pr{\"u}{\ss}, I.~Reuter, and F.~Schacherer, ``Transfac: an integrated
  system for gene expression regulation,'' \emph{Nucleic acids research},
  vol.~28, no.~1, pp. 316--319, 2000.

\bibitem{int19}
S.~M. Gallo, D.~T. Gerrard, D.~Miner, M.~Simich, B.~Des~Soye, C.~M. Bergman,
  and M.~S. Halfon, ``Redfly v3. 0: toward a comprehensive database of
  transcriptional regulatory elements in drosophila,'' \emph{Nucleic acids
  research}, vol.~39, no. suppl 1, pp. D118--D123, 2011.

\bibitem{int20}
M.~I. Barrasa, P.~Vaglio, F.~Cavasino, L.~Jacotot, and A.~J. Walhout, ``Edgedb:
  a transcription factor-dna interaction database for the analysis of c.
  elegans differential gene expression,'' \emph{BMC genomics}, vol.~8, no.~1,
  p.~21, 2007.

\end{thebibliography}
\end{document}